\documentclass[10pt,a4paper]{article}
\usepackage[a4paper, total={6in, 8in}]{geometry}
\makeatother

\makeatother
\usepackage{tabu}
\usepackage{amsmath}
\usepackage{amsthm}
\usepackage{breqn}
\usepackage{verbatim}
\usepackage{graphicx}
\usepackage{esvect}
\usepackage{cite}
\usepackage{amssymb}
\usepackage{footnote}
\usepackage{color}
\usepackage{url}
\usepackage{subcaption}
\usepackage{xcolor}
\usepackage{transparent}
\newtheorem{theorem}{Theorem}[section]
\newtheorem{corollary}{Corollary}[section]
\newtheorem{definition}{Definition}[section]
\newtheorem{lemma}[theorem]{Lemma}
\newtheorem{proposition}{Proposition}[section]
\newtheorem{remark}{Remark}
\newtheorem*{problem statement}{Problem Statement}
\newtheorem{feasibility problem}{Problem}
\newtheorem*{problem}{Problem}
\newtheorem{assumption}{Assumption}
\definecolor{mygreen}{RGB}{0,108,0}
\definecolor{myred}{RGB}{108,0,0}
\graphicspath{{./Figures/}}
\newcommand{\real}{\mathbb{R}}
\newcommand{\reg}{\mathcal{R}}
\newcommand{\setz}{\mathrm{Z}}
\newcommand{\sete}{\mathrm{E}}

\begin{document}
	\title{Isochronous Partitions for Region-Based Self-Triggered Control}
	\author{Giannis Delimpaltadakis and Manuel Mazo Jr.\thanks{The authors are with the Delft Center for Systems and Control, Delft University of Technology, Delft 2628CD, The Netherlands. Emails:\texttt{\{i.delimpaltadakis, m.mazo\}@tudelft.nl}. This work is supported by the ERC Starting Grant SENTIENT (755953).}}
	\date{}
	\maketitle
	\graphicspath{{./Figures2/}}
	\begin{abstract}
		In this work, we propose a \textit{region-based} self-triggered control (STC) scheme for nonlinear systems. The state space is partitioned into a finite number of regions, each of which is associated to a uniform inter-event time. The controller, at each sampling time instant, checks to which region does the current state belong, and correspondingly decides the next sampling time instant. To derive the regions along with their corresponding inter-event times, we use approximations of \textit{isochronous manifolds}, a notion firstly introduced in \cite{anta2012isochrony}. This work addresses some theoretical issues of \cite{anta2012isochrony} and proposes an effective computational approach that generates approximations of isochronous manifolds, thus enabling the region-based STC scheme. The efficiency of both our theoretical results and the proposed algorithm are demonstrated through simulation examples.
	\end{abstract}
	\section{Introduction}
	Control laws are, most often, implemented in a periodic fashion. However, despite periodic implementations facilitating controller design, they lead to overconsumption of available resources. Especially in Networked Control Systems (NCS) such implementations are considered inefficient, due to potential limitations on communication bandwidth. The need for resource-friendly control implementations has shifted the research focus to aperiodic schemes, namely Event-Triggered Control (ETC) \cite{arzen1999, heemels1999,tabuada2007etc,mazo2011decentralized, girard2015dynamicetc,antunes2015etc,lunze2011etc,lunze2010state} and Self-Triggered Control (STC) \cite{velasco2003self,tosample,anta2012isochrony,italy2013stcnonlinear,kallej2014stcnonlinear,mazo2010stciss,mazo2014stclinear,tolic2012stc,heemels2014stclq,fiter2012state,wang2009stc,wang2010stc,theodosis2018self}. For an introduction to STC and ETC see \cite{2012introtoetc_stc}. 
	
	These strategies assume sample-and-hold implementations, in which the control action is updated when a certain performance-related condition (\textit{triggering condition}) is satisfied. Triggering conditions are of the form $\phi(\zeta(t))\geq0$, where $\phi(\zeta(t))$ is a function of the state of the system, namely the \textit{triggering function}, e.g. see \cite{tabuada2007etc,girard2015dynamicetc}. Specifically in ETC, dedicated intelligent hardware constantly monitors the plant and detects when the triggering condition is satisfied. To relax this constraint, researchers have proposed STC as an alternative, according to which the controller predicts at each sampling time instant the next time at which the triggering condition would be satisfied. In this way, both ETC and STC promise to reduce the number of communication packets' transmissions and controller updates, thus saving both bandwidth and energy.
	
	Regarding STC for nonlinear systems, the amount of published work is limited. In \cite{tosample} the authors derive STC formulas employing interesting properties of homogeneous systems. Based on these properties, a different STC formula is proposed in \cite{anta2012isochrony}, employing the notion of isochronous manifolds.  In \cite{italy2013stcnonlinear}, a Taylor expansion of the Lyapunov function is used to predict the triggering times. In \cite{tolic2012stc} a self-triggered scheme is derived, based on a small-gain approach. In \cite{kallej2014stcnonlinear} a triggering condition that guarantees uniform ultimate boundedness for perturbed nonlinear systems is presented, and a corresponding self-triggered sampler is derived. Finally, the work in \cite{theodosis2018self} designs an STC scheme that copes with actuator delays.
	
	The STC formula proposed in \cite{tosample} proves to be conservative, i.e. it leads to a large amount of updates, at least when compared to the technique proposed here. This argument is illustrated in one of the simulation examples later in the document. What is more, the authors of \cite{theodosis2018self} admit that, although it addresses actuator delays, it is even more conservative than \cite{tosample}. Regarding \cite{anta2012isochrony} there are certain theoretical and practical issues, which are presented later in the introduction and are thoroughly discussed in this document. An important drawback of the rest of the STC techniques is that they require heavy computations that need to be carried out online. 
	
	A clever way to provide a trade-off between online computations and the number of updates in STC has already been proposed for linear systems with state feedback in  \cite{fiter2012state}. In particular, the authors in \cite{fiter2012state} discretize the state space of a linear system into a finite number of regions, assigning a particular self-triggered inter-event time to each region that lower bounds the event-triggered inter-event times of all points contained in that region. The computation of the self-triggered inter-event time for each region is carried out offline. Finally, in real-time the controller checks to which region of the state space does the current state belong and assigns to it the inter-event time of the corresponding region. To the best of our knowledge, there are no similar results for nonlinear systems.
	
	Motivated by the advantages of \cite{fiter2012state}, in this work we derive a \textit{region-based} STC scheme for nonlinear systems. In contrast to \cite{fiter2012state}, in which the state space is firstly discretized and afterwards the corresponding self-triggered inter-event times are computed, we propose to firstly predefine a set of specific inter-event times and afterwards derive the regions that correspond to the selected times. Thus, in our approach the number of regions in the state space is always equal to the number of times. This renders our approach more efficient and tames the curse of dimensionality, as the number of regions is independent of the dimensions of the system. 
	
	Towards discretizing the state space of nonlinear systems, we elaborate on the notion of \textit{isochronous manifolds}, originally introduced in \cite{anta2012isochrony}. Isochronous manifolds are hypersurfaces in the state space, that consist of points associated to the same inter-event time $\tau$, i.e. if the system's state belongs to an isochronous manifold at a sampling time $t_i$, then the next sampling time instant is $t_{i+1}=t_i+\tau$. In \cite{anta2012isochrony}, the authors propose a method to approximate these manifolds by upper-bounding the evolution of the triggering function, and then use the approximations to derive an STC formula. Unfortunately, there are some unaddressed theoretical and practical issues therein, which render the approximations, in general, invalid and hinder the application of the corresponding STC scheme. In particular, the bounding lemma presented in \cite{anta2012isochrony} (Lemma V.2 in \cite{anta2012isochrony}), based on which the upper-bounds of the triggering function are derived, is incorrect. Furthermore, we show that, even if a valid bound is obtained, the method proposed in \cite{anta2012isochrony} actually approximates the \textit{zero-level sets of the triggering function}, and not the actual isochronous manifolds. Finally, although the authors in \cite{anta2012isochrony} propose the use of SOSTOOLS \cite{sostools} to derive the approximations, we have found it to be numerically non-robust regarding solving this particular problem. 
	
	This paper tackles all of the aforementioned issues, in order to derive a discretization of the state space for nonlinear systems that enables a region-based STC scheme. Overall, the contributions of our work are the following:
	\begin{itemize}
		\item We present a valid version of the bounding lemma, based on a higher order comparison lemma \cite{comparison1971}.
		\item Employing this new lemma, we propose a refined methodology to approximate the actual isochronous manifolds of nonlinear ETC systems.
		\item We adjust a counter-example guided iterative method (see e.g. \cite{counterexample}) combining Linear Programming and SMT (Satisfiability Modulo Theory) solvers (e.g. \cite{dreal}), to derive an alternative algorithm that effectively computes approximations of isochronous manifolds.
		\item We derive a novel region-based STC scheme that provides a framework to trade-off online computational load with the number of updates. 
	\end{itemize}
	Finally, it is worth noting that isochronous manifolds are an inherent characteristic of any system with an output. Thus, as in \cite{anta2012isochrony}, the theoretical contribution of deriving approximations of isochronous manifolds might even exceed the context in which this paper is written.	
	
	\section{Notation and Preliminaries}
	\subsection{Notation}
	We denote points in $\mathbb{R}^n$ as $x$ and their Euclidean norm as $|x|$. We use the symbol $\exists!$, to denote existence and uniqueness. For $x,y\in\mathbb{R}^n$, we write $x\preceq y$ if $x_i\leq y_i$ ($i=1,\dots,n$), where the subscript $i$ denotes the $i$-th component of the corresponding vector. When there is no harm from ambiguity, the subscript $i$ may be, also, used to denote different points $x_i\in\mathbb{R}^n$.
	
	If $f:\mathbb{R}^n\rightarrow\mathbb{R}^m$ is $p$-times continuously differentiable, we write $f \in \mathcal{C}^p$. Let $X:M\rightarrow TM$ be a vector field and $h:M\rightarrow \mathbb{R}$ a map. $\mathcal{L}_Xh(x)$ denotes the Lie derivative of $h$ at a point $x$ along the flow of $X$. Similarly, $\mathcal{L}_X^k h(x)= \mathcal{L}_X(\mathcal{L}_X^{k-1} h(x))$ is the $k$-th Lie derivative with $\mathcal{L}_X^0 h(x) = h(x)$.
	
	Consider a system of first order differential equations:
	\begin{equation} \label{ode}
		\dot{\zeta}(t) = f(t, \zeta(t)).
	\end{equation}
	The solution of \eqref{ode} with initial condition $\zeta_0$ and initial time $t_0$ is denoted as $\zeta(t;t_0,\zeta_0)$. When $t_0$ (and $\zeta_0$) is clear from the context, then it is omitted, i.e. we write $\zeta(t;\zeta_0)$ ($\zeta(t)$).
	\subsection{Event-Triggered Control Systems}
	Consider a nonlinear control system:
	\begin{equation} \label{ct system}
		\dot{\zeta}(t) = f(\zeta(t),\upsilon(\zeta(t))),
	\end{equation}
	where $\zeta: \mathbb{R} \rightarrow \mathbb{R}^n$, $f:\mathbb{R}^n\times\mathbb{R}^m \rightarrow \mathbb{R}^n$, and a feedback control law $\upsilon: \mathbb{R}^n \rightarrow \mathbb{R}^m$. 
	A sample-and-hold implementation of \eqref{ct system} is typically applied by sampling the state of the system $\zeta(t)$ at time instants $t_i$, $i=0,1,2,...$, evaluating the input $\upsilon(\zeta(t_i))$ and keeping it constant until the next sampling time:
	\begin{equation*}
		\dot{\zeta}(t) = f(\zeta(t),\upsilon(\zeta(t_i))), \quad t\in[t_i, t_{i+1}).
	\end{equation*}
	We define the measurement error $\varepsilon(t)$ as the difference between the last measured state and the current state:
	\begin{equation}
		\varepsilon(t):=\zeta(t_i)-\zeta(t), \quad t\in[t_i, t_{i+1}).
		\label{measurement error}
	\end{equation}
	As soon as the updated control input is applied at each sampling time $t=t_i$, the state is measured and the error becomes 0, since $\zeta(t)=\zeta(t_i)$. With this definition, the sample-and-hold closed loop becomes:
	\begin{equation}
		\dot{\zeta}(t) = f(\zeta(t),\upsilon(\varepsilon(t)+\zeta(t))).
		\label{sample and hold system}
	\end{equation}
	
	In ETC the sampling time instants, or \textit{triggering times}, are defined as follows:
	\begin{equation} \label{triggering condition}
		t_{i+1} = t_{i} + \mathrm{inf}\{t>0: \phi(\zeta(t;x_{i}), \varepsilon(t;0))=0\}
	\end{equation}
	and $t_0=0$, where $x_{i}$ corresponds to the last measurement of the state of the plant. We call \eqref{triggering condition} the \textit{triggering condition},  $\phi(\cdot,\cdot)$ the  \textit{triggering function}, and the difference $t_{i+1}-t_i$ \textit{inter-event time}. Each point $x_i$ in the state space of the system, corresponds to a specific inter-event time denoted by $\tau(x_i)$:
	\begin{equation} \label{triggering time definition}
		\tau(x_i) := \mathrm{inf}\{t>0: \phi(\zeta(t;x_{i}), \varepsilon(t;0))= 0\}.
	\end{equation}
	During the interval $[t_i,t_{i+1})$, the triggering function starts from a negative value and remains negative until $t_{i+1}$. At $t_{i+1}$ it becomes zero. Typically, it is designed such that $\phi(\zeta(t;x_{i}), \varepsilon(t;0))\leq 0$ implies certain stability guarantees for the system. This justifies the choice \eqref{triggering condition} of sampling times.
	
	If we consider the extended state vector $\xi(t)=$$\begin{bmatrix}
		\zeta^{\top}(t)
		&\varepsilon^{\top}(t)
	\end{bmatrix} ^{\top}$$\in\real^{2n}$,
	the ETC system is written in a compact way:
	\begin{equation} \label{etc system}
		\begin{aligned}
			&\dot{\xi}(t)= \begin{bmatrix} f(\zeta(t),\upsilon(\zeta(t)+\varepsilon(t))\\
				-f(\zeta(t),\upsilon(\zeta(t)+\varepsilon(t)) \end{bmatrix}=F(\xi(t)), \text{ } t \in [t_i, t_{i+1}),\\
			&\xi_1(t_{i+1}^+)=\xi_1(t_{i+1}^-),\\
			&\xi_2(t_{i+1}^+)=0.
		\end{aligned}
	\end{equation}
	\begin{remark}
		Our analysis is carried out within the time interval $[0,t_{i+1}-t_i)=[0,\tau(x_i))$. Due to time-invariance of $F(\cdot)$, $\phi(\cdot)$, this is equivalent to analyzing within the interval $[t_i,t_{i+1})$.
	\end{remark}
	
	At any sampling time $t_i$, the state of \eqref{etc system} becomes $\xi(t_i)=(\zeta(t_i),0)=(x_i,0)$. Since we consider intervals between two sampling times, we focus on solutions $\xi(t;\xi_i)$ with $\xi_i=(x_i,0)$. Thus, we adopt the abusive notation $\phi(\xi(t;x_i))$, $\tau(x_i)$ (or later $\psi(x_i,t)$, $\mu(x_i,t)$) instead of $\phi(\xi(t;\xi_i))$, $\tau(\xi_i)$.
	\subsection{Self-Triggered Implementation}
	As aforementioned, self-triggered implementations remove the need for continuous monitoring of the triggering condition \eqref{triggering condition}, by predicting events $\phi(\xi(t;x))=0$. Specifically, an STC strategy dictates the next sampling time according to a function $\tau^{\downarrow}:\mathbb{R}^n\rightarrow\mathbb{R}^+$ lower-bounding the ETC inter-event times:
	\begin{equation} \label{stc bounds etc}
		\tau^{\downarrow}(x) \leq \tau(x).
	\end{equation}
	Since $\phi(\xi(t;x))<0$ for all $t\in[0,\tau(x))$, then it is guaranteed that $\phi(\xi(t;x))<0$ for all $t\in[0,\tau^{\downarrow}(x))$, and the stability of the system is preserved. Consequently, the STC inter-event times should be no larger than the corresponding ETC times in order to guarantee stability, but as large as possible in order to achieve greater reduction of updates. Finally, $\tau^{\downarrow}(\cdot)$ should be designed such that $\tau^{\downarrow}(x)\geq\epsilon>0$ for all $x$ in the operating region of the system, in order to avoid the scenario of infinite transmissions in finite amount of time (Zeno phenomenon).
	
	\section{Problem Statement}
	Inspired by \cite{fiter2012state}, the goal of this paper is to design a region-based STC scheme for nonlinear systems, providing a framework for trade-off between online computations and updates. In a region-based STC scheme, the state-space of the original system \eqref{sample and hold system} is divided into a finite number of regions $\mathcal{R}_i\in\mathbb{R}^n$ ($i=1,2,\dots$), each of which is associated to a self-triggered inter-event time $\tau_i$ such that:
	\begin{equation} \label{regions-times}
		\forall x \in \mathcal{R}_i:\quad \tau_i\leq\tau(x),
	\end{equation}
	where $\tau(x)$ denotes the event-triggered inter-event time associated to $x$ (see \eqref{triggering time definition}). The STC scheme operates as follows:
	\begin{enumerate}
		\item Measure the current state $\xi(t_k)=(x_k,0)$.
		\item Check to which of the regions $\mathcal{R}_i$ does $x_k$ belong.
		\item If $x_k\in\mathcal{R}_i$, set the next sampling time to $t_{k+1} = t_k + \tau_i$.
	\end{enumerate}
	The STC scheme preserves stability of the system, since the STC inter-event times lower bound the ETC ones (see \eqref{regions-times}).
	
	In \cite{fiter2012state} the state-space is discretized into regions $\mathcal{R}_i$ \textit{a-priori}, and afterwards the times $\tau_i$ are computed such that they satisfy \eqref{regions-times}. However, we propose an alternative approach: firstly a finite set of times $\{\tau_1,\tau_2,\dots\,\tau_q\}$ is \textit{predefined} (e.g. by the user), which will serve as STC inter-event times, with $\tau_i<\tau_{i+1}$, and then regions $\mathcal{R}_i$ corresponding to times $\tau_i$ are derived \textit{a-posteriori}, such that \eqref{regions-times} is satisfied. In this way, the number of regions is equal to the number of times $\tau_i$, in contrast to \cite{fiter2012state}, and the curse of dimensionality is tamed, as the number of regions does not depend on the system's dimensionality. Thus, the problem statement is as follows:
	\begin{problem statement}
		Given a finite set of times $\{\tau_1,\dots\,\tau_q\}$, with $\tau_i<\tau_{i+1}$ and $q>1$, find $\mathcal{R}_i\in\mathbb{R}^n$ that satisfy \eqref{regions-times}.
	\end{problem statement}
	Note that Zeno behaviour is ruled out by construction, since the STC inter-event times are lower bounded: $\tau^{\downarrow}(x)\geq\min_{i}\{\tau_i\}=\tau_1$. The choice of times $\tau_i$ and its effect is discussed later in the document.
	\section{Isochronous Manifolds, Triggering Level Sets and Discretization}
	
	Here, we recall results from \cite{anta2012isochrony} regarding isochronous manifolds, we introduce the notion of \textit{triggering level sets} and describe how isochronous manifolds and triggering level sets are different. Finally, we point out how, given proper approximations of isochronous manifolds, a state-space discretization state space is generated, enabling a region-based STC scheme.
	
	\subsection{Homogeneous Systems and Scaling of Inter-Event Times}
	First, we briefly go through some definitions regarding homogeneous functions and systems, and results previously derived in \cite{tosample} regarding scaling laws for inter-event times of homogeneous systems. Regarding the former, only the classical notion of homogeneity is presented. For the general definition of homogeneity, the reader is referred to \cite{homogeneity}.
	\begin{definition}[Homogeneous Function \cite{anta2012isochrony}]\label{homogeneous function}
		A function $f: \mathbb{R}^n \rightarrow \mathbb{R}^m$ is homogeneous of degree $\alpha \in \mathbb{N}$, if there exist $r_i>0$ ($i=1,2,\dots,m$) such that for all $x\in\real^n$:
		\begin{equation*}
			f_i(\lambda^{r_1}x_1, \dots, \lambda^{r_n}x_n) = \lambda^{\alpha+r_i}f_i(x_1,\dots,x_n), \quad \forall \lambda>0,
		\end{equation*}
		where $f_i(x)$ is the $i$-th component of $f(x)$ and $\alpha>-\min_{i}r_i$.	
	\end{definition}
	\begin{definition}[Homogeneous System]
		A system \eqref{ct system} is called homogeneous of degree $\alpha \in \mathbb{R}$, whenever $f(\zeta(t),\upsilon(\zeta(t)))=\tilde {f}(\zeta(t))$ is a homogeneous function of the same degree.
	\end{definition} 
	
	We now review the scaling laws of inter-event times previously derived in \cite{tosample}. Along lines passing through the origin (but excluding the origin) the event-triggered inter-event times scale according to the following rule:
	\begin{theorem}[Scaling Law \cite{tosample}]
		Consider a dynamical system \eqref{etc system} homogeneous of degree $\alpha$ and a triggering function $\phi(\cdot)$ homogeneous of degree $\theta$. For all $x\in\real^n$, the inter-event times $\tau:\mathbb{R}^{n} \rightarrow \mathbb{R}^+\cup \{+\infty\}$ defined by \eqref{triggering time definition} scale as:
		\begin{equation} \label{scaling}
			\tau(\lambda x) = \lambda^{-\alpha} \tau(x), \quad \lambda>0.
		\end{equation}
	\end{theorem}
	In the following, we refer to lines going through the origin as \textit{homogeneous rays}. Notice that the scaling law for the inter-event times \eqref{scaling} does not depend on the degree of homogeneity of the triggering function considered. The property derives from the following useful lemma:
	\begin{lemma}[Time-Scaling Property \cite{tosample}]
		Consider an ETC system \eqref{etc system} and a triggering function $\phi(\cdot)$ homogeneous of degree $\alpha$ and $\theta$, respectively. The triggering function satisfies:
		\begin{equation} \label{trig_cond_scaling}
			\phi(\xi(t;\lambda x)) = \phi(\lambda\xi(\lambda^{\alpha}t;x))=  \lambda^{\theta+1}\phi(\xi(\lambda^{\alpha}t;x)),
		\end{equation}
		where the first equality is a property of homogeneous flows. 	
	\end{lemma}
	\begin{assumption} \label{assumption 1}
		For the remaining of the paper, our analysis is based on the following set of assumptions:
		\begin{itemize}
			\item The extended ETC system \eqref{etc system} is smooth and homogeneous of degree $\alpha \geq 1$, with $r_i=1$ for all $i$.
			\item The triggering function $\phi(\xi(t;x))$ is smooth and homogeneous of degree $\theta \geq 1$, with $r_i=1$ for all $i$.
			\item For all $x\in\mathbb{R}^n-\{0\}$, $\phi(\xi(0;x))<0$ and $\exists t_x\in(0,+\infty)$ such that $\phi(\xi(t_x;x))=0$.
			\item Compact sets $\setz\subset\real^n$ and $\Xi\subset\real^{2n}$, containing a neighbourhood of the origin, are given, such that for all $x\in \setz$: $\phi(\xi(t;x))\leq0\implies \xi(t;x)\in \Xi$.
			\item The system \eqref{ct system} has the origin as the only equilibrium.
		\end{itemize}
	\end{assumption}
	\begin{remark}
		The aforementioned analysis and Assumption \ref{assumption 1} constitute the framework within which this work is carried out. Nevertheless, as pointed out in \cite{anta2012isochrony} (Lemma IV.4 therein),  any smooth function can be rendered homogeneous, if embedded in a higher dimensional space. Thus, our results are applicable to general smooth nonlinear systems and triggering functions. This is thoroughly discussed in Appendix D and showcased in Section VII.B via a numerical example.
		\label{framework remark}
	\end{remark}
	\begin{remark}
		The set $\Xi$ could be $\Xi = \setz\times\sete$, where $\setz=\{x\in\real^n:\text{ }V(x)\leq c\}$, $\sete=\{x_0-x\in\real^n:\text{ }x,x_0\in\setz\}$, $c>0$, and $V(\cdot)$ is a radially unbounded Lyapunov function for the ETC system. In most ETC schemes (e.g. \cite{tabuada2007etc}), $V(x)$ is given and the triggering function satisfies:  $\phi(\xi(t;x))\leq0\implies\dot{V}(\zeta(t;x))\leq 0$. Thus, trajectories of \eqref{etc system} starting from $\setz\times\{0\}\subset\Xi$ stay in $\Xi=\setz\times\sete$. The intuition behind this assumption is analyzed in Section V.C. An alternative way of constructing $\setz$ and $\Xi$ is demonstrated in Section VII.B.
		\label{z=lyap remark}
	\end{remark}
	\subsection{Isochronous Manifolds and Triggering level Sets}
	\begin{definition}[Isochronous Manifolds]
		Consider a closed loop system \eqref{etc system} and a triggering function $\phi(\cdot)$. The set $M_{\tau_{\star}}=\{x\in\mathbb{R}^n : \tau(x)=\tau_{\star}\}$, where $\tau(x)$ is defined by \eqref{triggering time definition}, is called an isochronous manifold of time $\tau_{\star}$.
		\label{manifold definition}
	\end{definition}
	Alternatively, all points $x\in \mathbb{R}^n$ which correspond to inter-event time $\tau_{\star}$ constitute the isochronous manifold $M_{\tau_{\star}}$. Isochronous manifolds are of dimension $n-1$ (proven in \cite{anta2012isochrony}). 
	\begin{definition}[Triggering Level Sets]
		We call the set:
		\begin{equation}
			L_{\tau_{\star}} := \{x\in \mathbb{R}^n: \phi(\xi(\tau_{\star};x))=0\}
			\label{level set}
		\end{equation}
		triggering level set of $\phi(\xi(\tau_{\star};x))$ for time $\tau_{\star}$.
	\end{definition}
	Triggering level sets are the zero-level sets of the triggering function, for fixed $t$.
	Let us now make a crucial observation: \emph{The equation $\phi(\xi(t;x))=0$ may have multiple solutions with respect to time $t$ for a given $x$.} In other words, there might exist points $x \in \mathbb{R}^n$ and time instants $\tau_{x,1}<\tau_{x,2}<...<\tau_{x,k}$, with $k>1$ such that $\phi(\xi(\tau_{x,i};x))=0$ for all $i=1, 2, ..., k$. We briefly present an example with a triggering function exhibiting multiple zero-crossings for given initial conditions:
	
	\textit{Example:} Consider the jet-engine compressor control system from \cite{jet_paper}:
	\begin{equation*}
		\dot{\xi}_1(t) = - \xi_2(t)- \frac{3}{2}\xi^2_1(t) - \frac{1}{2}\xi^3_1(t), \quad
		\dot{\xi}_2(t) = \upsilon(\xi(t)),
	\end{equation*}
	with control law $\upsilon(\xi(t))=\xi_1(t) -\frac{1}{2}(\xi_1^2(t)+1)(y + \xi_1^2(t)y +\xi_1(t)y^2) +2\xi_1(t)$, where $y=2\tfrac{\xi_1^2+\xi_2}{\xi_1^2+1}$. A triggering function that guarantees asymptotic stability is the following \cite{tosample}:
	\begin{equation*}
		\phi(\xi(t;x)) = |\varepsilon|^2-0.82\sigma^2|\xi(t;x)|^2 \text{,} \qquad \sigma \in (0,1).
	\end{equation*}
	\begin{figure}[!h]
		\centering
		\includegraphics[width=2in]{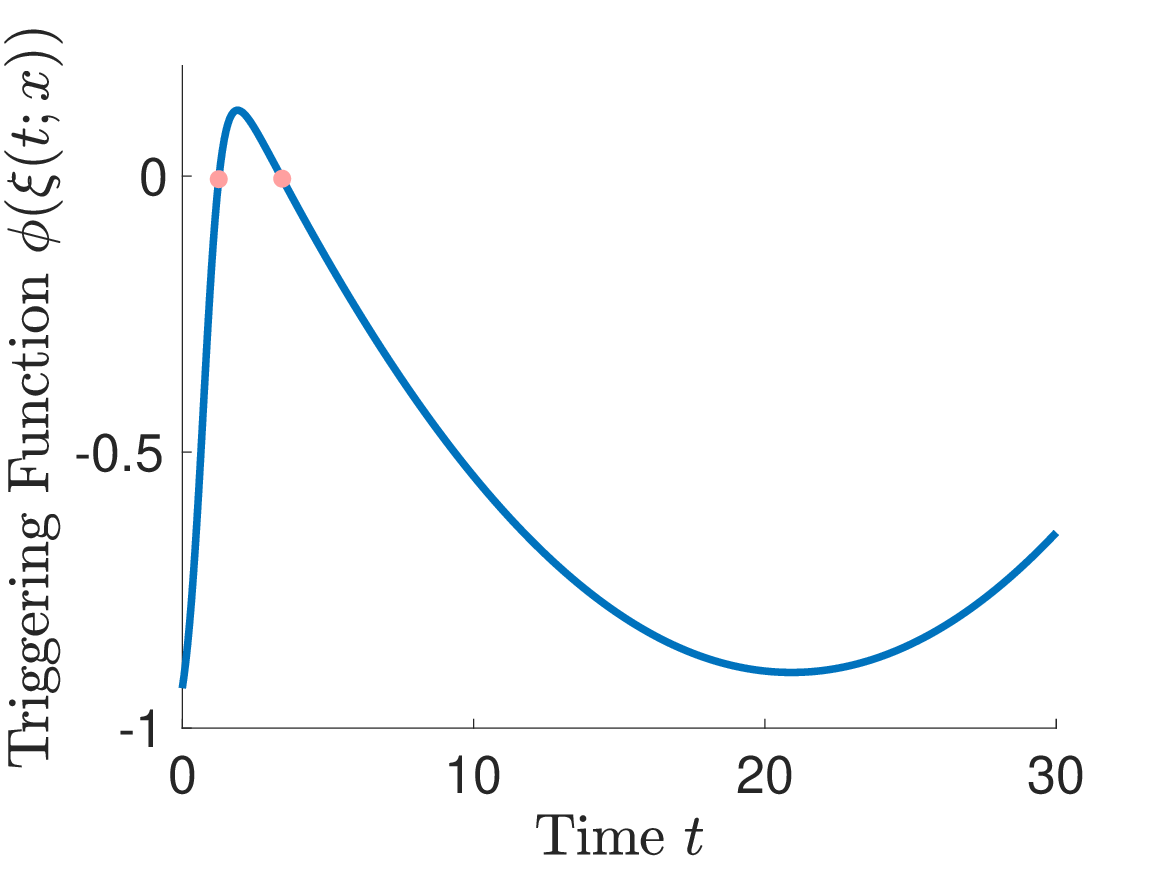}
		\caption{The time evolution of $\phi(x;t)$ for initial condition $[-0.5, -1]^{\top}$. It exhibits multiple zero-crossings.}
		\label{jet_trig}
	\end{figure} 
	The evolution of the triggering function $\phi(\xi(t;x))$ for the initial condition $[-0.5 -1]^{\top}$ is simulated and illustrated in Fig. \ref{jet_trig}. It is clear from the figure that it exhibits multiple zero-crossings, for $t=\tau_{x,1}\approx 1.15s$ and $t=\tau_{x,2}\approx 3.22s$.\hfill$\blacksquare$ 
	
	Inter-event times are defined as \textit{the first zero-crossing of the triggering function} (see \eqref{triggering time definition}), i.e. $\tau(x)=\tau_{x,1}$. Isochronous manifolds are defined with respect to this first zero-crossing, and any point $x \in \mathbb{R}^n-\{0\}$ belongs only to one isochronous manifold: $M_{\tau_{x,1}}$. However, the same point belongs to all triggering level sets $L_{\tau_{x,i}}$. For instance, in the previous example, the point $x=(-0.5, -1)$ belongs to both triggering level sets $L_{1.15}$ and $L_{3.22}$, whereas it belongs to only one isochronous manifold, i.e. $M_{1.15}$. In \cite{anta2012isochrony}, isochronous manifolds and triggering level sets are treated as if they were identical, which creates problems regarding approximating isochronous manifolds. This is addressed later in the document.
	
	\begin{remark} \label{one zero-crossing remark}
		If the triggering function $\phi(\xi(t;x))$ has only one zero-crossing for all $x\in\mathbb{R}^n-\{0\}$, then the triggering level sets do coincide with the isochronous manifolds, i.e. $M_{\tau_{\star}}=\{x\in\mathbb{R}^n : \tau(x)=\tau_{\star}\}= \{x\in \mathbb{R}^n: \phi(\xi(\tau_{\star};x))=0\} = L_{\tau_{\star}}$.
	\end{remark}
	
	Isochronous manifolds possess the two following properties:
	\begin{proposition}[\hspace{1sp}\cite{anta2012isochrony}]\label{star manifolds proposition}
		Consider an ETC system \eqref{etc system}, a triggering function $\phi(\cdot)$, and let Assumption \ref{assumption 1} hold. Each homogeneous ray intersects any isochronous manifold only at one point:
		\begin{equation} \label{star manifolds}
			\forall \tau_{\star}>0 \text{ and }\forall x \in \mathbb{R}^n-\{0\}: \text{ } \exists! \lambda_x>0 \text{ such that } \lambda_x x \in M_{\tau_{\star}}
		\end{equation}
	\end{proposition}
	\begin{proof}
		According to \eqref{scaling} and \eqref{trig_cond_scaling}, on any homogeneous ray, times vary from $0$ to $+\infty$ as $\lambda_x$ varies from $+\infty$ to $0$. Thus, for any $\tau_{\star}\in \mathbb{R}^+$ there exists a point $x$ on each ray such that $\tau(x)=\tau_{\star}$. What is more, equation \eqref{scaling} implies that there do not exist two different points on the same homogeneous ray that correspond to the same inter-event time.
	\end{proof}
	\begin{proposition} \label{encircle prop}
		Consider an ETC system \eqref{etc system}, a triggering function $\phi(\cdot)$, and let Assumption \ref{assumption 1} hold. Consider isochronous manifolds $M_{\tau_i}$ and $M_{\tau_{i+1}}$, with $\tau_{i}<\tau_{i+1}$. The following holds for all $x \in M_{\tau_i}$:
		\begin{equation} \label{encirclement}
			\exists! \lambda_x\in(0,1) \text{ s.t. } \lambda_x x \in M_{\tau_{i+1}} \wedge \not\exists \kappa_x\geq 1 \text{ s.t. } \kappa_x x \in M_{\tau_{i+1}}.
		\end{equation}
	\end{proposition}
	\begin{proof}
		According to Proposition \ref{star manifolds proposition}, since each homogeneous ray intersects any isochronous manifold only at one point, $\exists! \lambda_x>0$ such that $\lambda_xx\in M_{\tau_{i+1}}$, where $x\in M_{\tau_i}$. From the scaling law \eqref{scaling} we get:
		\begin{equation*}
			\tau_{i+1}=\tau(\lambda_x x) = \lambda_x^{-\alpha}\tau_{i}\implies \lambda_x = \sqrt[\alpha]{(\tfrac{\tau_{i}}{\tau_{i+1}})}<1,
		\end{equation*}
		since $\tau_i<\tau_{i+1}$. There can be no other intersection of the homogeneous ray with $M_{\tau_{i+1}}$, i.e. $\not\exists \kappa_x\geq 1 \text{ s.t. } \kappa_x x \in M_{\tau_{i+1}}$.
	\end{proof}
	Proposition \ref{encircle prop} states that isochronous manifolds for smaller times are further away from the origin. Given \eqref{star manifolds}, in Fig. \ref{star and non star} the curve on the top could be an isochronous manifold of a homogeneous system, while the two bottom curves cannot.
	\begin{figure}[!h]
		\centering
		\includegraphics[width=2.5in]{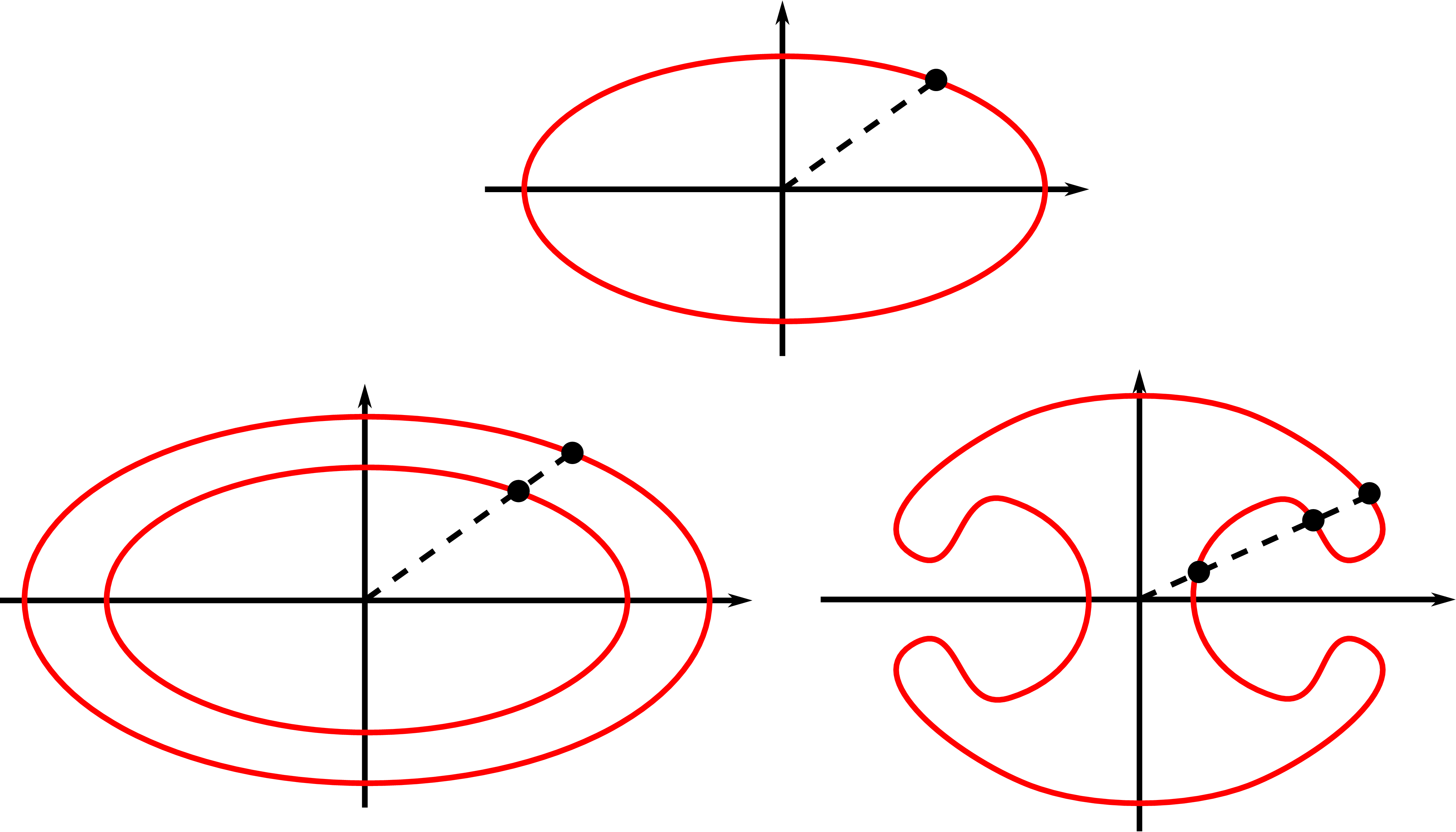}
		\caption{The curve on the top is intersected only once by each homogeneous ray, thus it could be an isochronous manifold of a homogeneous system. The two bottom curves are intersected by some homogeneous rays more than once, thus they cannot be isochronous manifolds of a homogeneous system.}
		\label{star and non star}
	\end{figure}
	\begin{remark} \label{homogeneity remark}
		Properties \eqref{star manifolds} and \eqref{encirclement} of isochronous manifolds result directly from the time scaling property \eqref{trig_cond_scaling}. 
	\end{remark}		
	
	\subsection{State-Space Discretization and a Self-Triggered Strategy}
	
	For the following, we assume that the system operates in an arbitrarily large compact set $\mathcal{B}$ the whole time. Assume that isochronous manifolds $M_{\tau_i}$ for $\tau_1<\tau_2<\tau_3$ are given, as illustrated in Fig. \ref{discr example}. We define the regions between manifolds as: 
	\begin{equation}\label{reg_between_mans}
		\begin{aligned}
			R_i = \{x\in \mathbb{R}^n: \exists \kappa_x\geq1 \text{ s.t. } \kappa_x x \in M_{\tau_i} \wedge\exists \lambda_x\in(0,1) \text{ s.t. } \lambda_x x \in M_{\tau_{i+1}}\}
		\end{aligned}
	\end{equation}
	for $\tau_i<\tau_{i+1}$, and the region enclosed by the manifold $M_{\tau_3}$ as $R_3 = \{x\in \mathbb{R}^n: \exists \kappa_x\geq1 \text{ s.t. } k_x x \in M_{\tau_3}\}$. Since \eqref{encirclement} holds, a region $R_i$ is the set with its outer boundary being $M_{\tau_i}$ and its inner boundary being $M_{\tau_{i+1}}$. 
	\begin{figure}[!h]
		\centering
		\includegraphics[width=2in]{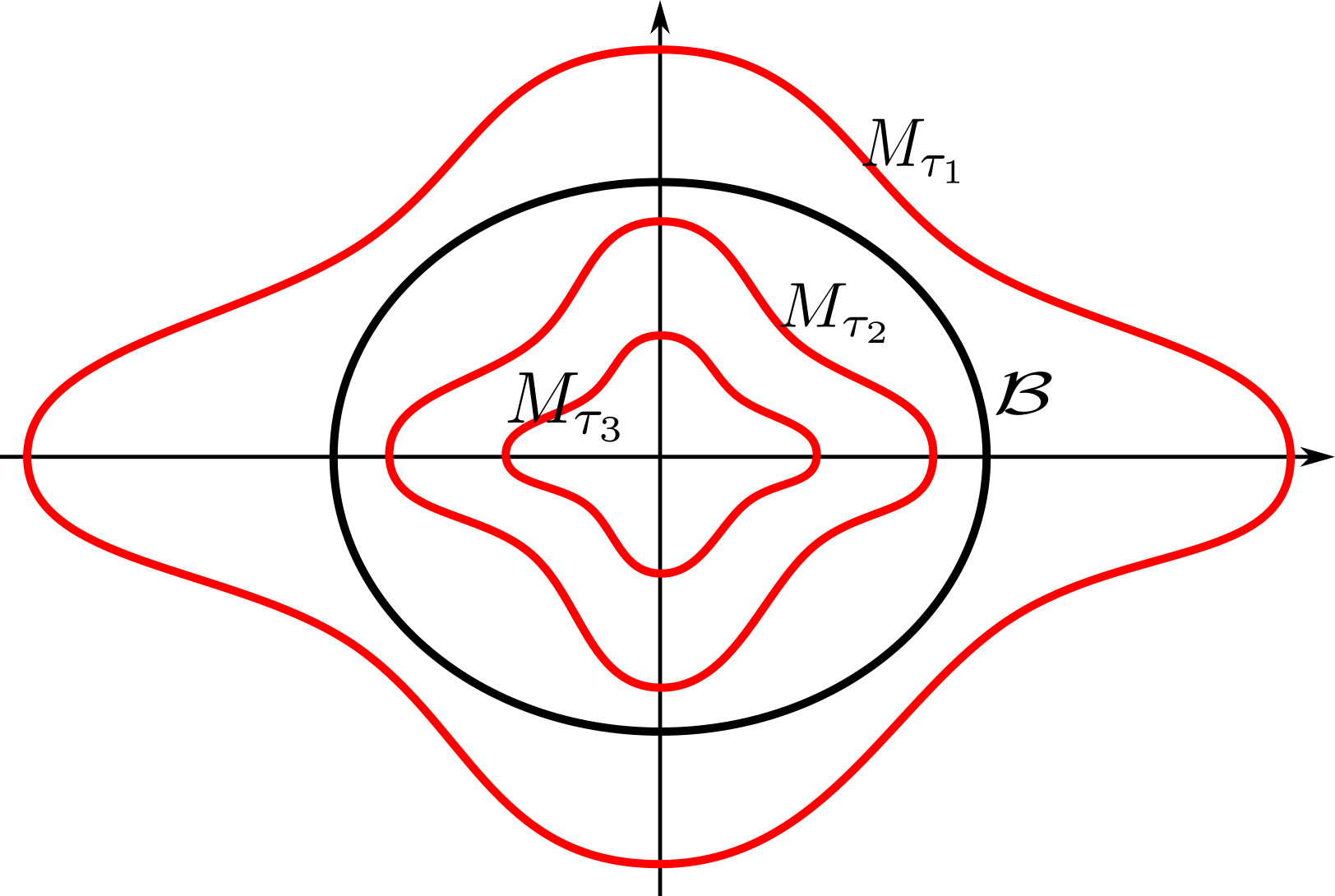}
		\caption{Isochronous manifolds $M_{\tau_1}$, $M_{\tau_2}$, $M_{\tau_3}$ (red lines) for $\tau_1<\tau_2<\tau_3$, and the operating region $\mathcal{B}$ (black line).}
		\label{discr example}
	\end{figure}
	The scaling law \eqref{scaling} implies that: $\tau(x)\geq\tau_i$ for all $x\in R_i$. Thus, isochronous manifolds could be employed for discretizing the state space in regions $R_i$ such that \eqref{regions-times} is satisfied. If isochronous manifolds did not satisfy property \eqref{star manifolds}, then the regions $R_i$ could potentially intersect with each other (see Fig. \ref{bad discr}). Hence, it would not be possible to derive a discretization as the one described.
	\begin{figure}[!h]
		\centering
		\includegraphics[width=1in]{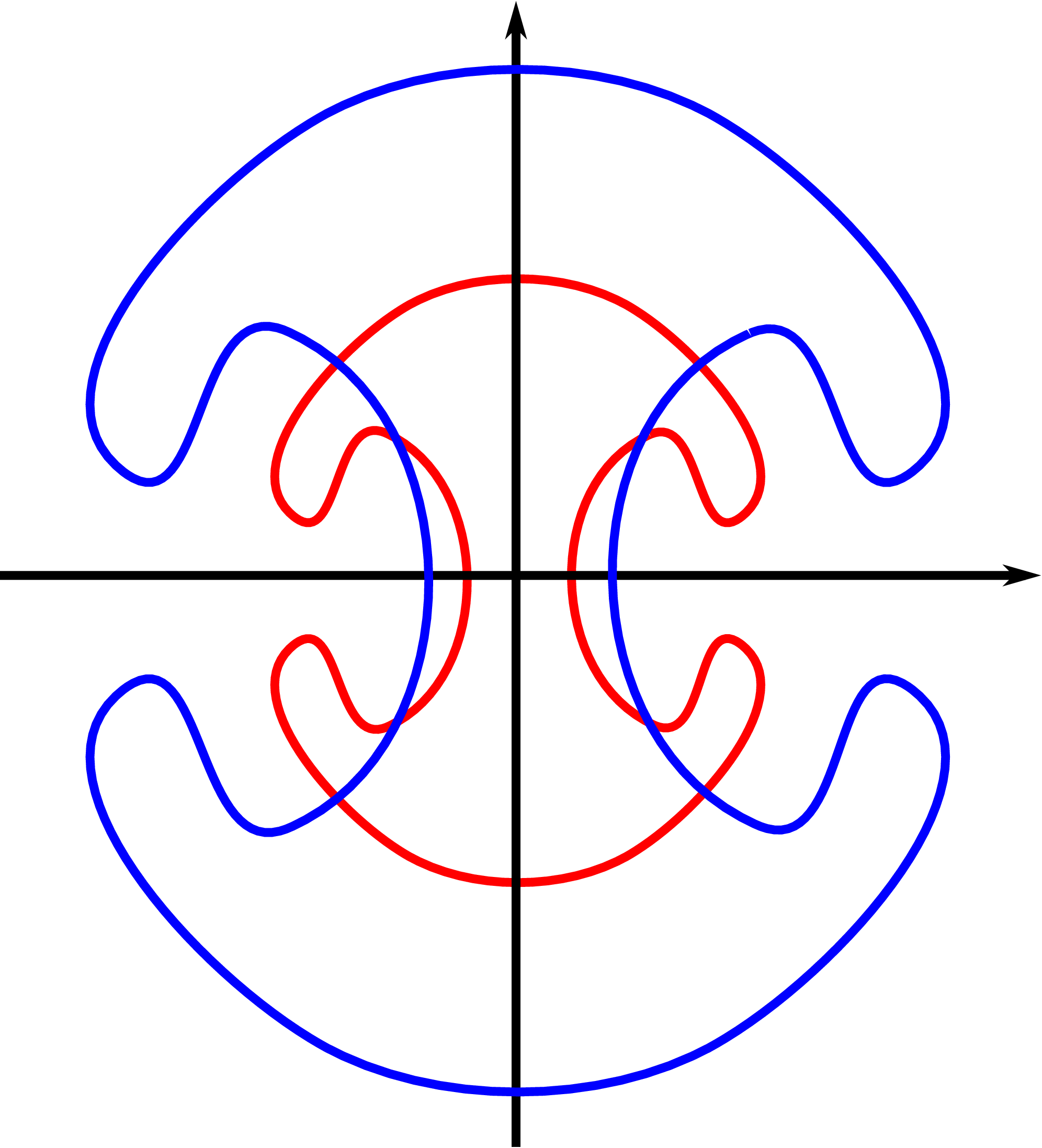}
		\caption{If isochronous manifolds did not satisfy \eqref{star manifolds}, it would not be possible to discretize the state space enabling a region-based STC scheme.}
		\label{bad discr}
	\end{figure}	
	
	\subsection{Inner-Approximations of Isochronous Manifolds and Discretization}
	Deriving the actual isochronous manifolds is generally not possible, as nonlinear systems most often do not admit a closed-form analytical solution. Thus, in order to discretize the state space and generate a region-based STC scheme, we propose a method to construct inner-approximations of isochronous manifolds, as shown in Fig. \ref{discr approx}.
	\begin{definition}[Inner-Approximations of Isochronous Manifolds]
		Consider a system \eqref{etc system} and a triggering function, and let Assumption \ref{assumption 1} hold. A set $\underline{M}_{\tau_i}$ is called inner approximation of an isochronous manifold $M_{\tau_i}$ if and only if for all $x \in \underline{M}_{\tau_i}$:
		\begin{equation}
			\exists \kappa_x\geq 1 \text{ s.t. } \kappa_x x \in M_{\tau_i} \text{ } \mathrm{and} \text{ } \not\exists \lambda_x\in(0,1) \text{ s.t. } \lambda_x x \in M_{\tau_i}.
		\end{equation}
	\end{definition}
	\begin{figure}[!h]
		\centering
		\includegraphics[width=2in]{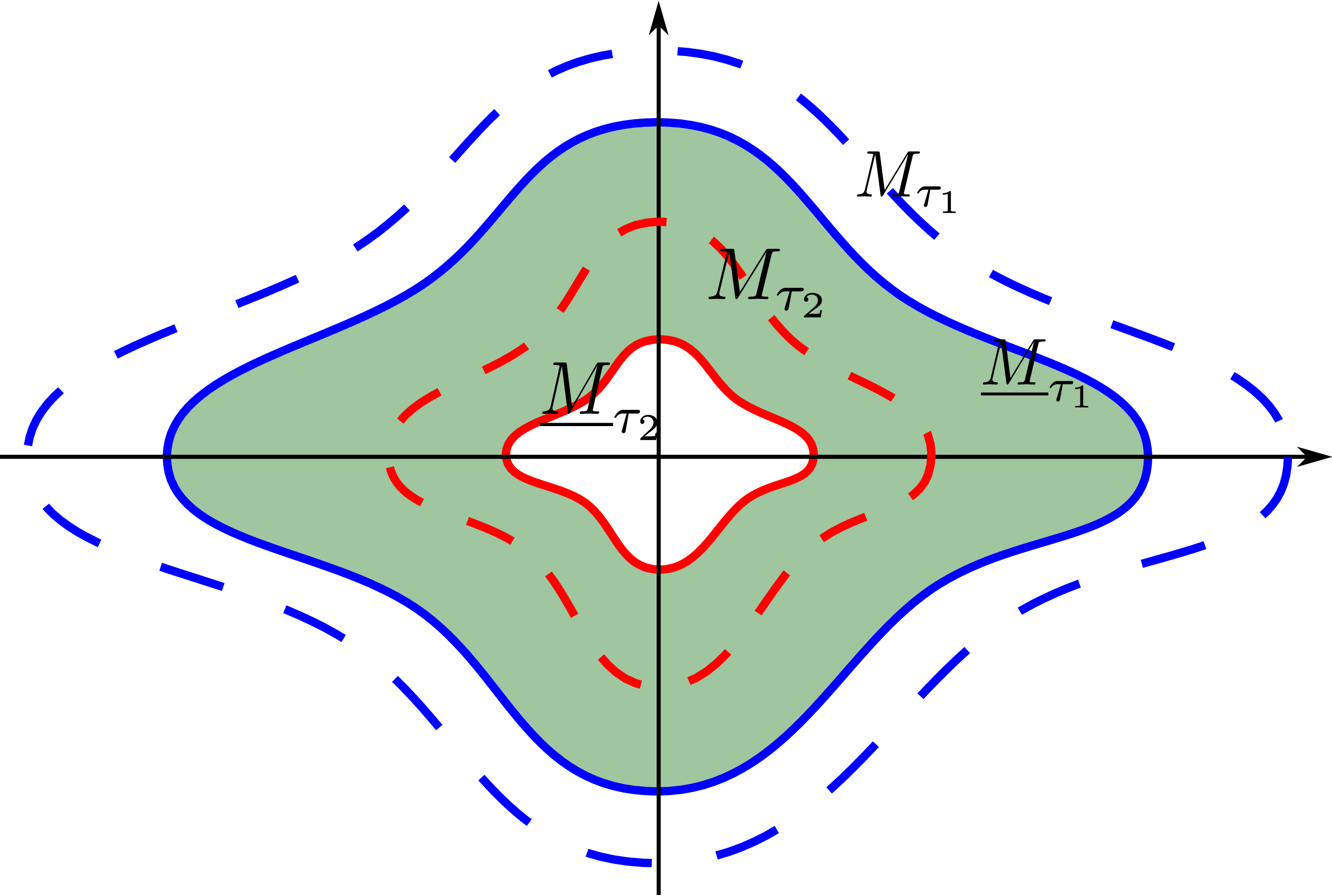}
		\caption{Isochronous manifolds $M_{\tau_i}$ (dashed lines), and their inner-approximations $\underline{M}_{\tau_i}$ (solid lines). The filled region represents $\reg_1$.}
		\label{discr approx}
	\end{figure}
	In other words, an inner-approximation of an isochronous manifold is contained inside the region encompassed by the isochronous manifold. Consider inner-approximations $\underline{M}_{\tau_i}$ of isochronous manifolds ($\tau_1<\tau_2<...$), that satisfy properties \eqref{star manifolds} and \eqref{encirclement}. We consider the regions between sets $\underline{M}_{\tau_i}$:
	\begin{equation} \label{approx region}
		\begin{aligned}
			\reg_i = \{x\in \mathbb{R}^n: \exists \kappa_x\geq1 \text{ s.t. } \kappa_x x \in \underline{M}_{\tau_i} \wedge\exists \lambda_x\in(0,1) \text{ s.t. } \lambda_x x \in \underline{M}_{\tau_{i+1}}\}.
		\end{aligned}
	\end{equation}
	A region $\reg_i$ is the set with its outer boundary being $\underline{M}_{\tau_i}$ and its inner boundary being $\underline{M}_{\tau_{i+1}}$ (see Fig. \ref{discr approx}). For such sets, by \eqref{scaling} we get the following result:
	\begin{corollary}\label{corollary}
		Consider a system \eqref{etc system} and a triggering function $\phi(\cdot)$, and let Assumption \ref{assumption 1} hold. Consider two inner-approximations $\underline{M}_{\tau_{i}}$ and $\underline{M}_{\tau_{i+1}}$ of isochronous manifolds, with $\tau_i\leq\tau_{i+1}$. Assume that $\underline{M}_{\tau_{i}}$ and $\underline{M}_{\tau_{i+1}}$ satisfy \eqref{star manifolds} and \eqref{encirclement}. For the region $\reg_i$ defined in \eqref{approx region}, the following holds:
		\begin{equation*}
			\forall x \in \reg_{i}:  \tau_{i}\leq \tau(x).
		\end{equation*}
	\end{corollary}
	\begin{proof}
		For all $x\in\reg_i$, $\exists \kappa_x\geq1$ s.t. $\kappa_xx\in\underline{M}_{\tau_i}$. Thus, $\exists k_x\geq\kappa_x\geq1$ s.t. $k_xx\in M_{\tau_i}$. By \eqref{scaling}, we have $\tau(k_xx)=\tau_i\implies \tau(x)=k_x^{\alpha}\tau_i\geq\tau_i$.
	\end{proof}
	Thus, given inner-approximations of isochronous manifolds, the state space can be discretized into regions $\reg_{i}$, enabling the region-based STC scheme. \emph{This construction requires that inner approximations should also satisfy \eqref{star manifolds} and \eqref{encirclement}. Deriving inner-approximations $\underline{M}_{\tau_{\star}}$ of isochronous manifolds such that they satisfy \eqref{star manifolds} and \eqref{encirclement} constitutes the main theoretical challenge of this work.}
	\begin{remark}\label{number of regions remark}
		As already noted, the number of regions $\reg_i$ equals the number $q$ of predefined times $\tau_i$ (see Section III). Given that $\tau_1$ and $\tau_q$ are fixed, as the number of times $q$ grows, the areas of regions $\reg_i$ become smaller, as the same space is discretized into more regions. Thus, the STC inter-event times $\tau_i$ become more accurate bounds of the actual ETC times $\tau(x)$. However, during the online implementation, the controller in general needs to perform more checks to determine the region of a measured state. Hence, the number $q$ of times $\tau_i$ provides a trade-off between computations and conservativeness.
	\end{remark}
	\begin{remark}\label{first manifold remark}
		Note that $\tau_1$ has to be selected, such that the operating region $\mathcal{B}$ lies completely inside the region delimited by $\underline{M}_{\tau_1}$ (e.g. see Fig. \ref{discr example}). To check this, the approach of \cite{tosample} or an SMT (Satisfiability Modulo Theory) solver (e.g. \cite{dreal}) can be used.
	\end{remark}
	\begin{remark}
		For non-homogeneous systems, there will always exist a neighbourhood around the origin that cannot be contained in any region $\reg_{i}$. However, this set can be made arbitrarily small, by selecting a sufficiently small time $\tau_1$. For a thorough discussion on this, the reader is referred to Appendix D.
	\end{remark}
	
	\section{Approximations of Isochronous Manifolds}
	Here a refined methodology is presented, which generates inner-approximations of isochronous manifolds that satisfy \eqref{star manifolds} and \eqref{encirclement}. First, we show how the method of \cite{anta2012isochrony} actually approximates triggering level sets, and then we refine its core idea to derive approximations of isochronous manifolds. 
	\subsection{Approximations of Triggering Level Sets}
	The method proposed in \cite{anta2012isochrony} is based on bounding the time evolution of the triggering function by another function with linear dynamics: $\psi_1(x,t) \ge \phi(\xi(t;x))$, with $\psi_1(x,0)=\phi(\xi(0;x))<0$ for all $x \in \mathbb{R}^n-\{0\}$. The bound is obtained by constructing a linear system according to a bounding lemma (Lemma V.2 in \cite{anta2012isochrony}). Unfortunately, this lemma is invalid and the function that is obtained does not always bound $\phi(\xi(t;x))$. Specifically, a counterexample is given in \cite{lemma_counterexample} (pp.2 Example 2). However, later in the document we present a slightly adjusted lemma, that is actually valid. Thus, for this subsection we assume that $\psi_1(x,t)$ is an upper bound of $\phi(\xi(t;x))$.
	
	Since $\psi_1(x,t) \ge \phi(\xi(t;x))$ and $\psi_1(x,0)<0$, if we define:
	\begin{equation*}
		\tau^{\downarrow}(x) = \inf\{t>0 : \psi_1(x,t)=0\},
	\end{equation*}
	then it is guaranteed that $\phi(\xi(x;t))\leq 0, \quad \forall t\in [0,\tau^{\downarrow}(x)]$.
	Hence, the first zero-crossing of $\psi_1(x,t)$ for a given $x$ happens before the first zero-crossing of $\phi(\xi(t;x))$, i.e. the inter-event time of $x$ is lower bounded by $\tau^{\downarrow}(x)$: $\tau(x) \geq \tau^{\downarrow}(x)$.
	
	In \cite{anta2012isochrony}, under the misconception that isochronous manifolds and triggering level sets coincide, it is argued that to approximate an isochronous manifold, it suffices to approximate the set $L_{\tau_{\star}}:=\{x\in \mathbb{R}^n: \phi(\xi(\tau_{\star};x))=0\}$, i.e. a triggering level set. Thus, the upper bound $\psi_1(x,t)$ of $\phi(\xi(t;x))$ is used to derive the following approximation: $\underline{L}_{\tau_{\star}}:=\{x\in \mathbb{R}^n: \psi_1(x,\tau_{\star})=0\}$. 
	However, as we have already pointed out for the triggering function, $\psi_1(x,t)$ might also have multiple zero-crossings for a given $x\in\mathbb{R}^n$. Thus, the equation $\psi_1(x,t)=0$ does not only capture the inter-event times of points $x$, but possibly also more zero-crossings of $\phi(t;x)$. Thus, we can say that the set $\underline{L}_{\tau_{\star}}$ is an approximation of the triggering level set $L_{\tau_{\star}}$, and not of the isochronous manifold $M_{\tau_{\star}}$. Furthermore, observe that $\psi_1(x,t)$ does not a-priori satisfy the time scaling property \eqref{trig_cond_scaling}. Consequently, there is no formal guarantee that the sets $\underline{L}_{\tau_{\star}}$ satisfy \eqref{star manifolds} (see Remark \ref{homogeneity remark}). In other words, the sets $\underline{L}_{\tau_{\star}}$ might be intersected by some homogeneous rays more than once, or they may not be intersected at all.
	\begin{remark}\label{stc of anta remark}
		In \cite{anta2012isochrony}, given a fixed time $\tau_{\star}$, the equation
		\begin{equation}\label{find lambda}
			\psi_1(\tfrac{x_0}{\lambda},\tau_{\star})=0
		\end{equation}
		is solved w.r.t. $\lambda$, in order to determine the STC inter-event time of the measured state $x_0$ as: $\tau^{\downarrow}(x_0)=\lambda^{-\alpha}\tau_{\star}$. Note that \eqref{find lambda} finds intersections $\tfrac{x_0}{\lambda}$ of $\underline{L}_{\tau_{\star}}$ with the ray passing through $x_0$. Hence, the above observations imply that \eqref{find lambda} may not have any real solution, or may admit some solutions $\lambda$ such that $\tau^{\downarrow}(x_0)=\lambda^{-\alpha}\tau_{\star}>\tau(x)$, hindering stability.
	\end{remark}
	
	\subsection{Inner-Approximations of Isochronous Manifolds}
	Although, the method of \cite{anta2012isochrony} generates approximations of triggering level sets, which do not satisfy \eqref{star manifolds}, we employ the idea of upper-bounding the triggering function, and we impose additional properties to the upper bound, such that the obtained sets approximate isochronous manifolds and satisfy \eqref{star manifolds} and \eqref{encirclement}. Remarks \ref{one zero-crossing remark} and \ref{homogeneity remark} state that: 1) isochronous manifolds coincide with triggering level sets, if $\phi(\cdot)$ has only one zero-crossing w.r.t. $t$, and 2) $\phi(\cdot)$ satisfying \eqref{trig_cond_scaling} implies that isochronous manifolds satisfy \eqref{star manifolds} and \eqref{encirclement}. Intuitively, we could construct a function $\mu(x,t)$ that satisfies the same properties and its zero-crossing happens before the one of $\phi(\cdot)$, and use the level sets $\underline{M}_{\tau_{\star}} = \{x\in\mathbb{R}^n: \mu(x,\tau_{\star})=0\}$ as inner approximations of isochronous manifolds that satisfy \eqref{star manifolds} and \eqref{encirclement}. The above are summarized in the following theorem:
	\begin{theorem} \label{theorem 1}
		Consider an ETC system \eqref{etc system}, a triggering function $\phi(\cdot)$, and let Assumption \ref{assumption 1} hold. Let $\mu:\mathbb{R}^n\times\mathbb{R}^+\rightarrow\mathbb{R}$ be a function that satisfies:
		\begin{subequations}
			\begin{align}
				&\mu(x,0)<0, \quad \forall x \in \mathbb{R}^n-\{0\}, \label{bound init cond}\\
				&\mu(x,t) \geq \phi(\xi(t;x)), \quad \forall t \in [0,\tau(x)] \text{ }\mathrm{and} \text{ }\forall x \in \mathbb{R}^n-\{0\}, \label{bound time validity}\\
				&\mu(\lambda x,t) = \lambda^{\theta+1}\mu(x,\lambda^{\alpha}t), \quad \forall t,\lambda>0 \text{ }\mathrm{and}\text{ } \forall x \in \mathbb{R}^n-\{0\}, \label{scaling of bound}\\
				&\forall x \in \mathbb{R}^n-\{0\}: \quad \exists!\tau_x  \text{ such that } \mu(x,\tau_x)=0. \label{one-zero-crossing of bound}
			\end{align}
			\label{bound requirements}
		\end{subequations}
		The sets $\underline{M}_{\tau_{\star}} = \{x\in\mathbb{R}^n: \mu(x,\tau_{\star})=0\}$ are inner-approximations of isochronous manifolds $M_{\tau_{\star}}$ and satisfy \eqref{star manifolds} and \eqref{encirclement}.
	\end{theorem}
	\begin{proof}
		See Appendix.
	\end{proof}
	\begin{remark}\label{bound time validity remark}
		It is crucial that inequality \eqref{bound time validity} extends at least until $\tau(x)$, in order for $\mu(x,t)$ to capture the actual inter-event time, i.e. for the minimum time satisfying $\mu(x,t)=0$ to lower bound the minimum time satisfying $\phi(\xi(t;x))=0$.
	\end{remark}
	
	\subsection{Constructing the Upper Bound of the Triggering Function}
	In this subsection, we construct a valid bounding lemma and we employ it in order to derive an upper bound $\mu(x,t)$ of the triggering function $\phi(\xi(t;x))$, such that it satisfies \eqref{bound requirements}.
	\begin{lemma} 
		Consider a system of differential equations $\dot{\xi}(t) = F(\xi(t))$, where $\xi:\mathbb{R}^+\rightarrow\mathbb{R}^n$, $F:\mathbb{R}^n \rightarrow \mathbb{R}^n$, a function $\phi:\mathbb{R}^n \rightarrow \mathbb{R}$ and a set $\Omega_d = \{x \in \mathbb{R}^n: |x|< d \}$. For every set of coefficients $\delta_{0}, \delta_{1}, ..., \delta_p \in \mathbb{R}^+$ satisfying:
		\begin{equation}
			\mathcal{L}_F^p\phi(z) \le \sum_{i=0}^{p-1}\delta_i\mathcal{L}_F^i\phi(z) + \delta_p, \quad \forall z \in \Omega_d,
			\label{delta ineq}
		\end{equation}
		the following inequality holds for all $\xi_0 \in \Omega_d$:
		\begin{equation*}
			\phi(\xi(t;\xi_0)) \le \psi_1(y(\xi_0),t) \quad \forall t \in [0,\tau_{\xi_0}),
		\end{equation*}
		where $\tau_{\xi_0}$ is defined as: 
		\begin{equation}
			\tau_{\xi_0}=\sup\{\tau>0:\xi(t;\xi_0) \in \Omega_d, \quad \forall t \in [0,\tau)\}
			\label{tau xi0}
		\end{equation}
		and $\psi_1(y(\xi_0),t)$ is the first component of the solution of the following linear dynamical system:
		\begin{equation}
			\dot{\psi}= 
			\begin{bmatrix}
				0 &1 &0 &\dots &0 &0\\
				0 &0 &1 &\dots &0 &0\\
				\vdots &\vdots &\qquad &\ddots &\vdots &\vdots\\
				0 &0 &0 &\dots &1 &0\\
				\delta_0 &\delta_1 &\delta_2 &\dots &\delta_{p-1} &1\\
				0 &0 &0 &\dots &0  &0\\
			\end{bmatrix} \psi = A\psi,
			\label{linear system}
		\end{equation}
		with initial condition: 
		\begin{equation*}
			y(\xi_0)=\begin{bmatrix}\phi(\xi_0) &\mathcal{L}_F\phi(\xi_0) &\dots &\mathcal{L}_F^{p-1}\phi(\xi_0) &\delta_p\end{bmatrix}^{\top}.
		\end{equation*}
		\label{our lemma} 
	\end{lemma}
	\begin{proof}
		See Appendix.
	\end{proof}
	\begin{remark}\label{difference between lemmas}
		The main difference between Lemma \ref{our lemma} and the bounding lemma in \cite{anta2012isochrony} is that in Lemma \ref{our lemma} the coefficients $\delta_i$ are forced to be non-negative. We also include a proof, employing a higher-order comparison lemma, since the comparison lemma arguments used in the proof of \cite{anta2012isochrony} are invalid.
	\end{remark}
	
	Let us define the open ball:
	\begin{equation}
		\Omega_d:=\{x\in\mathbb{R}^{2n}: |x|<d\}. \label{omega xi}
	\end{equation} Consider the following feasibility problem:
	\begin{feasibility problem} \label{feasibility problem}
		Consider a system \eqref{etc system} and a triggering function $\phi(\cdot)$ and let Assumption \ref{assumption 1} hold. Find $\delta_0,\dots,\delta_p\in\real$ such that:
		\begin{subequations}
			\begin{align}	
				&\mathcal{L}_F^p\phi(z) \le \sum_{i=0}^{p-1}\delta_i\mathcal{L}_F^i\phi(z) + \delta_p, \quad \forall z\in \Omega_d, \label{delta ineq3}\\
				&\delta_0\phi\Big((x,0)\Big)+\delta_p \geq \epsilon > 0, \quad \forall x \in \setz, \label{delta init cond}\\
				&\delta_i \geq 0, \quad i=0,1,\dots,p, \label{delta positivity}
			\end{align}
			\label{delta feasibility}
		\end{subequations}
		where $\epsilon$ is an arbitrary predefined positive constant, $d$ is such that $\Xi\subset \Omega_d$, and $\setz$, $\Xi$ and $\Omega_d$ are given by Assumption \ref{assumption 1} and \eqref{omega xi} respectively.
	\end{feasibility problem}
	The feasible solutions of \eqref{delta feasibility} belong in a subset of the feasible solutions of Lemma \ref{our lemma}, i.e. the solutions of \eqref{delta feasibility} determine upper bounds of the triggering function. Moreover, such $\delta_i$ always exist, since to satisfy \eqref{delta feasibility} it suffices to pick $\delta_p \geq \max\{\epsilon,\sup\limits_{z \in\Omega_d}\mathcal{L}_F^p\phi(z)\}$ and $\delta_i=0$ for $i=0, \dots, p-1$. The following theorem shows how to employ solutions of Problem \ref{feasibility problem}, in order to construct upper bounds that satisfy \eqref{bound requirements}.
	\begin{theorem}\label{main theorem}
		Consider a system \eqref{etc system}, a triggering function $\phi(\cdot)$, and coefficients $\delta_0,\dots,\delta_p$ solving Problem \ref{feasibility problem}. Let Assumption \ref{assumption 1} hold. Let $D=\{x\in\mathbb{R}^n: |x|=r\}$, with $r>0$ and $D \subset \setz$. Define the following function for all $x\in\mathbb{R}^n-\{0\}$:
		\begin{equation} \label{mu}
			\mu(x,t) := C (\tfrac{|x|}{r})^{\theta+1} \boldsymbol{e}^{A(\frac{|x|}{r})^{\alpha}t} 
			\begin{bmatrix}
				\phi\Big((r\tfrac{x}{|x|},0)\Big)\\
				\max\bigg( \mathcal{L}_f\phi\Big((r\tfrac{x}{|x|},0)\Big),0 \bigg)\\
				\vdots\\
				\max\bigg( \mathcal{L}_f^{p-1}\phi\Big((r\tfrac{x}{|x|},0)\Big),0 \bigg)\\
				\delta_p
			\end{bmatrix}
		\end{equation}
		where $A$ is as in \eqref{linear system}, $C=[1\text{ } 0\text{ } \dots\text{ } 0]$, and $\alpha$ and $\theta$ are the degrees of homogeneity of the system and the triggering function, respectively. The function $\mu(x,t)$ satisfies \eqref{bound requirements}.
	\end{theorem}
	\begin{proof}
		See Appendix.
	\end{proof}
	Thus, according to Theorem \ref{theorem 1}, the sets $\underline{M}_{\tau_{\star}} = \{x\in\mathbb{R}^n:\mu(x,\tau_{\star})=0\}$ are inner-approximations of the actual isochronous manifolds of the system and satisfy \eqref{star manifolds} and \eqref{encirclement}. The fact that $\mu(x,t)$ satisfies \eqref{bound requirements} directly implies that the region $\reg_i$ between two approximations $\underline{M}_{\tau_i}$ and $\underline{M}_{\tau_{i+1}}$ ($\tau_i<\tau_{i+1}$) can be defined as:
	\begin{equation} \label{approx regions}
		\reg_i := \{x\in\mathbb{R}^n: \mu(x,\tau_i)\leq0 \wedge \mu(x,\tau_{i+1})>0\}.
	\end{equation}
	To determine online to which region does the measured state belong, the controller checks inequalities like the ones in \eqref{approx regions}.
	
	Let us explain the importance of $\setz$, $\Xi$ from Assumption \ref{assumption 1}. By solving Problem \ref{feasibility problem}, an upper bound $\psi(\xi_0,t)$ is determined according to Lemma V.2 that bounds $\phi(\xi(t;\xi_0))$ as follows:
	\begin{equation*}
		\psi(\xi_0,t) \geq \phi(\xi(t;\xi_0)), \quad \forall \xi_0\in\Omega_d \text{ and } \forall t\in[0,\tau_{\xi_0}),
	\end{equation*} 
	where $\tau_{\xi_0}$ is the time when the trajectory $\xi(t;\xi_0)$ leaves $\Omega_d$ (see \eqref{tau xi0}).
	\begin{figure}[!h]
		\centering
		\includegraphics[width=1.3in]{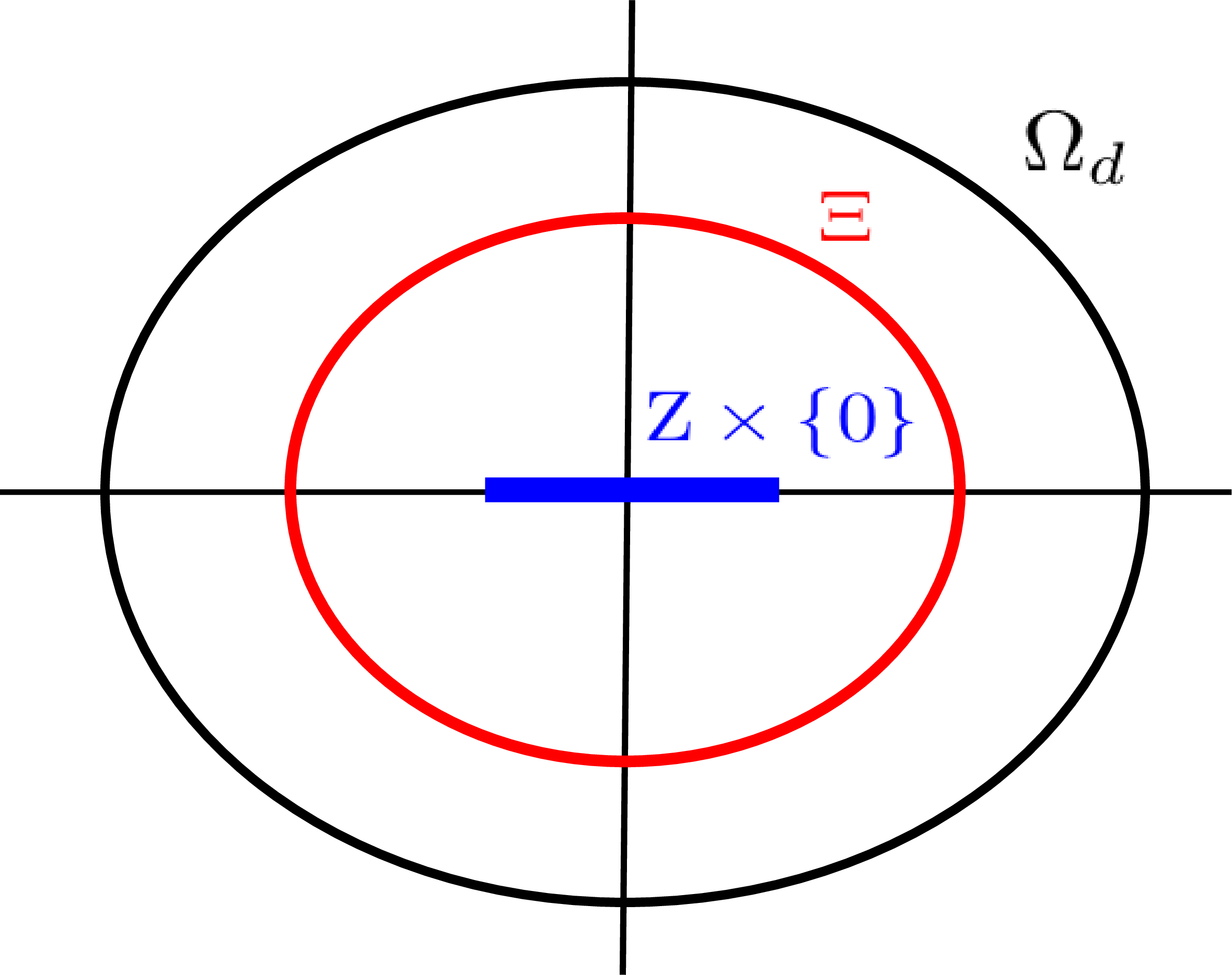}
		\caption{The sets $\mathrm{Z}\times\{0\}\subset\Xi\subset\Omega_d$.}
		\label{z_xi_omega}
	\end{figure}
	What is needed is to bound $\phi(\xi(t;\xi_0))$ at least until the inter-event time $\tau(\xi_0)$ (see Remark 9), i.e. $\tau(\xi_0)<\tau_{\xi_0}$. This is exactly what Assumption \ref{assumption 1} offers: trajectories starting from points $\xi_0\in\mathrm{Z}\times\{0\}$ stay in $\Xi\subset\Omega_d$ at least until $\tau(\xi_0)$ (see Figure \ref{z_xi_omega}). In other words, for all points $\xi_0\in\mathrm{Z}\times\{0\}$, we have that $\tau(\xi_0)<\tau_{\xi_0}$ (since $\Xi\subset\Omega_d$) and therefore:
	\begin{equation}\label{psi bounds phi}
		\psi(\xi_0,t) \geq \phi(\xi(t;\xi_0)), \quad \forall \xi_0\in\mathrm{Z}\times\{0\} \text{ and } \forall t\in[0,\tau(\xi_0)].
	\end{equation}
	Regarding the $\{0\}$-part of $\mathrm{Z}\times\{0\}$, note that we only consider initial conditions $\xi_0=(x,0)$, as aforementioned. Finally, transforming $\psi(x,t)$ into $\mu(x,t)$ by incorporating properties \eqref{scaling of bound} and \eqref{one-zero-crossing of bound}, equation \eqref{psi bounds phi} becomes \eqref{bound time validity}. All these statements are formally proven in the Appendix.
	
	\section{An Algorithm that Derives Upper Bounds}\label{algorithm section}
	Although in \cite{anta2012isochrony} SOSTOOLS \cite{sostools} is proposed for deriving the $\delta_{i}$ coefficients, our experience indicates that it is numerically non-robust regarding solving this particular problem. We present an alternative approach based on a Counter-Example Guided Iterative Algorithm (see e.g. \cite{counterexample}), which combines Linear Programming and SMT solvers (e.g. \cite{dreal}), i.e. tools that verify or disprove first-order logic formulas, like \eqref{delta feasibility}.
	
	Consider the following problem formulation:
	\begin{problem} 
		Find a vector of parameters $\Delta$ such that:
		\begin{equation}\label{inequality continuous domain}
			G(x)\cdot\Delta \leq b(x), \quad \forall x \in \Omega,
		\end{equation}
		where $\Delta \in \mathbb{R}^p$, $G:\mathbb{R}^n \rightarrow \mathbb{R}^{m\times p}$, $b:\mathbb{R}^n \rightarrow \mathbb{R}^m$ and $\Omega$ is a compact subset of $\mathbb{R}^n$. 
	\end{problem}
	For the initialization of the algorithm, a finite subset $\hat{\Omega}$ consisting of samples $x_i$ from the set $\Omega$ is obtained. Notice that the relation: $G(x_i)\cdot\Delta\leq b(x_i), \quad \forall x_i \in \hat{\Omega}$ can be formulated as a linear inequality constraint: $\hat{A}\cdot\Delta \leq \hat{b}$, where $\hat{A}=\begin{bmatrix}G^{\top}(x_1) &G^{\top}(x_2) &\dots &G^{\top}(x_i) &\dots \end{bmatrix}^{\top}$ and $\hat{b}=\begin{bmatrix}b(x_1) &b(x_2) &\dots &b(x_i) &\dots \end{bmatrix}^{\top},$ $\forall x_i \in \hat{\Omega}$.  Each iteration of the algorithm consists of the following steps:
	\begin{enumerate}
		\item Obtain a candidate solution $\hat{\Delta}$ by solving the following linear program (LP):
		\begin{equation*}
			\text{minimize} \quad \mathbf{c^{\top}\Delta}, \quad 
			\text{subject to} \quad \mathbf{\hat{A}\cdot\Delta\leq \hat{b}},
		\end{equation*}
		where $\mathbf{c}$ can be freely chosen by the user (we discuss meaningful choices later).
		\item Employing an SMT solver, check if the candidate solution $\hat{\Delta}$ satisfies the inequality on the original domain, i.e. if $G(x)\cdot\hat{\Delta} \leq b(x),$ $\forall x \in \Omega$:
		\begin{enumerate}
			\item If $\hat{\Delta}$ satisfies \eqref{inequality continuous domain}, then the algorithm terminates and returns $\hat{\Delta}$ as the solution.
			\item If $\hat{\Delta}$ does not satisfy \eqref{inequality continuous domain}, the SMT solver returns a point $x_c \in \Omega$ where this inequality is violated, i.e. a counter-example. Add $x_c$ to $\hat{\Omega}$ and update accordingly the matrices $\hat{A}$ and $\hat{b}$. Go to step 1.  
		\end{enumerate}
	\end{enumerate}
	
	Note that in step 2b) a single constraint is added to the LP of the previous step, i.e. $G(x_c)\cdot\Delta\leq b(x_c)$, by concatenating $G(x_c)$ and $b(x_c)$ to the $\hat{A}$ and $\hat{b}$ matrices, respectively.
	
	In order to solve Problem \ref{feasibility problem} in particular, we define $\Delta = \begin{bmatrix}
		\delta_0 &\delta_1 &\dots &\delta_p
	\end{bmatrix} ^{\top}$ and:
	\begin{align*}
		b(\cdot) &= \begin{bmatrix}
			-\mathcal{L}_F^p\phi(z)
			&-\varepsilon
			&\dots
			&0
		\end{bmatrix}^{\top}\\
		G(\cdot) &= \begin{bmatrix}
			-\phi(z) &\dots &-\mathcal{L}_F^{p-1}\phi(z) &-1\\
			-\phi(\xi(0;x_0)) &0 &\dots &-1\\
			-1 &0 &\dots &0\\
			0 &-1 &\dots &0\\
			0  &0 &\ddots  &0\\
			0 &0 &\dots &-1
		\end{bmatrix},
	\end{align*}
	where $z \in \Omega_d$ and $x_0 \in \setz$, with $\Omega_d$ and $\setz$ as in \eqref{omega xi} and Assumption \ref{assumption 1} respectively. Hence, the set $\hat{\Omega}$ consists of points $X_i=(z_i,x_{0_i}) \in \Omega_d\times \setz$, and after solving the corresponding LP, the SMT solver checks if $G(X)\cdot\hat{\Delta} \leq b(X),$ $\forall X \in \Omega_d\times \setz$. Finally, intuitively, tighter estimates of $\mathcal{L}_F^p\phi(z)$ could be obtained by minimizing $\delta_p$, and using the other $\mathcal{L}_F^i\phi(z)$ terms in the right hand side of \eqref{delta ineq}. Hence, $\mathbf{c}=\begin{bmatrix}
		0 &\dots &0 &1
	\end{bmatrix}$ constitutes a wise choice for the LP. In the following section, numerical examples demonstrate the algorithm's efficiency, alongside the validity of our theoretical results.
	\begin{remark}\label{parameters_remark}
		It is recommended that the parameter $d$, which determines the size of $\Omega_d$, is chosen relatively small, in order to help the algorithm terminate faster. Moreover, our experiments indicate 
		that just 2 initial samples $x_i\in\hat{\Omega}$ are sufficient for the algorithm to terminate relatively quickly. Intuitively, this is because letting the algorithm determine most of the samples itself (by finding the counter-example points) is more efficient than dictating samples a-priori. Finally, $p$ should be chosen large enough so that the obtained bound $\mu(\cdot,\cdot)$ is tight, but also small enough so that the dimensionality of the feasibility problem remains small. According to our experience, a choice of $2\leq p \leq 4$ leads to satisfactory results and quick termination of the algorithm, in most cases.
	\end{remark}
	
	\section{Simulation Results}
	In the following numerical examples SOSTOOLS failed to derive upper bounds, as it mistakenly reasoned that Problem \ref{feasibility problem} is infeasible. The upper bounds were derived by the algorithm proposed above.
	\subsection{Homogeneous System}
	In this example, we compare the region-based STC with the STC technique of \cite{tosample} (which is also computationally light) and with ETC (which constitutes the ideal scenario). Consider the following homogeneous control system:
	\begin{equation}
		\dot{\zeta}_1 = -\zeta_1^3+\zeta_1\zeta_2^2, \quad
		\dot{\zeta}_2 = \zeta_1\zeta_2^2 -\zeta_1^2\zeta_2 + \upsilon,
		\label{example1}
	\end{equation}
	with $\upsilon(\zeta)=-\zeta_2^3-\zeta_1\zeta_2^2$. A homogeneous triggering function for an asymptotically stable ETC implementation is:
	\begin{equation*}
		\phi(\xi(t;x)) = |\varepsilon(t;x)|^2-0.0127^2\sigma^2|\zeta(t;x)|^2, \quad \sigma\in(0,1),
	\end{equation*}
	where $\xi(\cdot)$ denotes the trajectories of the corresponding extended system \eqref{etc system}, $\varepsilon(\cdot)$ is the measurement error \eqref{measurement error}, and $x$ is the previously sampled state. As in \cite{anta2012isochrony}, we select $\sigma=0.3$.
	
	In order to test the proposed region-based STC scheme, Problem \ref{feasibility problem} is solved by employing the algorithm presented in the previous section. In particular, we set $p=3$, $\Omega_d=\{x\in\mathbb{R}^4: |x|<0.9\}$ and $\Xi=\setz\times\sete$, where $\setz=\{x\in\mathbb{R}^2: V(x)\leq 0.1\}$, $\sete=\{x_0-x\in\real^2:\text{ }x,x_0\in\setz\}$ and $V(x)=\frac{1}{2}x_1^2+\frac{1}{2}x_2^2$ is a Lyapunov function for the system. Observe that $\Xi\subset\Omega_d$. The coefficients found are
	$\delta_0 = 0$, $\delta_1=0.1272$, $\delta_2=0$ and $\delta_3=0.0191$.
	In order to construct $\mu(x,t)$ according to \eqref{mu}, we fix $r=0.29$ and the set $D=\{x\in\mathbb{R}^2: |x|=r\}$ indeed lies in the interior of $\setz$. The state space is discretized into 348 regions $\reg_i$ with corresponding self-triggered inter-event times $\tau_{348}=0.1s$ and $\tau_i=1.01^{-2}\tau_{i+1}$. Indicatively, 4 derived approximations of isochronous manifolds are shown in Fig. \ref{example1_discr}. Observe that the approximations satisfy \eqref{star manifolds} and \eqref{encirclement}.
	\begin{figure}[!h]
		\centering
		\includegraphics[width=2.4in]{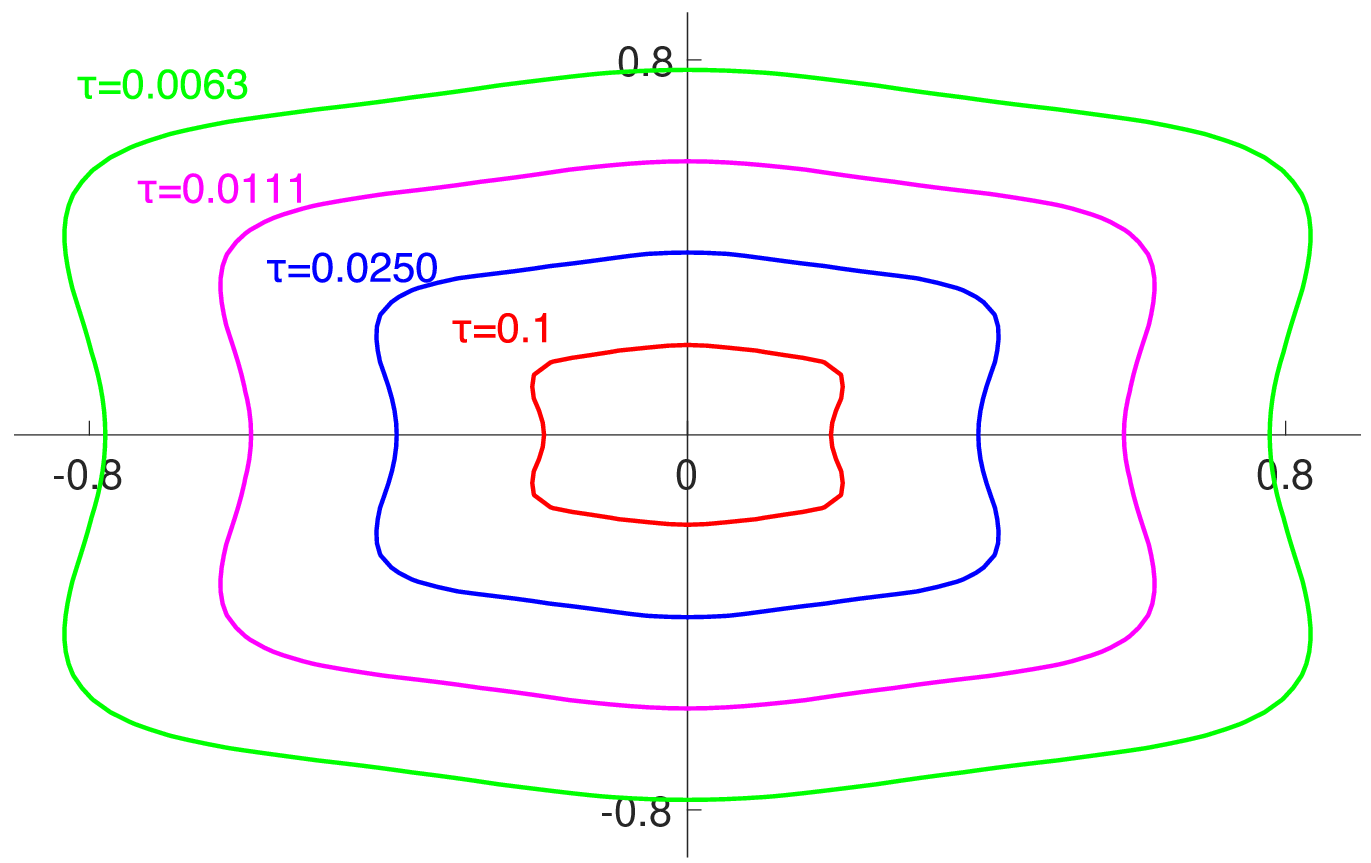}
		\caption{Approximations of isochronous manifolds of the ETC implementation of \eqref{example1}.}
		\label{example1_discr}
	\end{figure}
	
	The system is initiated at $x=[1,1]^{\top}$ and the simulation lasts for 5s. Fig. \ref{example1_stc_etc_times} compares the time evolution of the inter-event times of the region-based STC, the STC proposed in \cite{tosample} and ETC. In total, ETC triggered 383 times, the region-based STC triggered 554 times, whereas the STC of \cite{tosample} triggered 2082 times. 
	\begin{figure}[!h]
		\centering
		\includegraphics[width=3in]{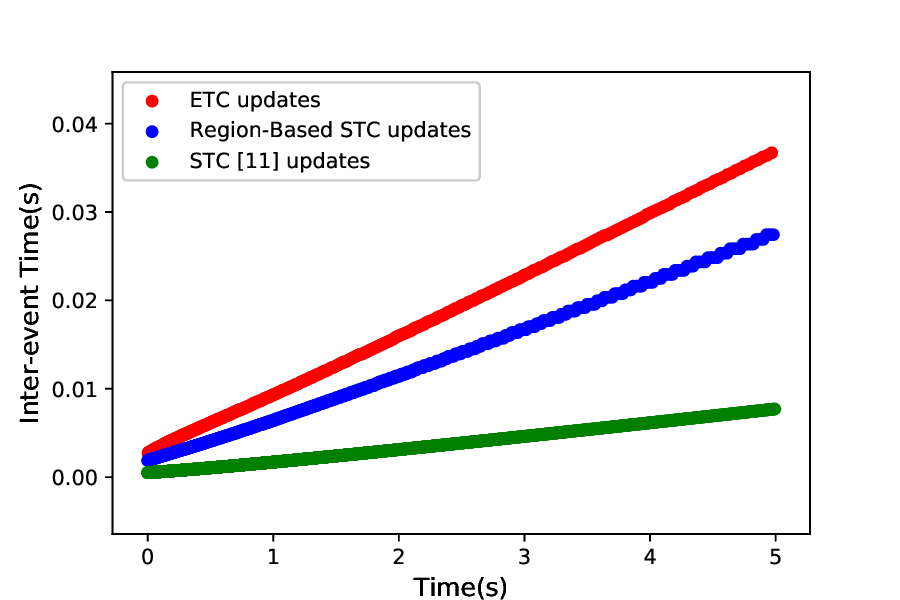}
		\caption{The time evolution of region-based STC, STC of \cite{tosample} and ETC inter-event times along the trajectory of \eqref{example1}.}
		\label{example1_stc_etc_times}
	\end{figure}
	Given Fig. \ref{example1_stc_etc_times} and the number of total updates for each technique we can conclude that: 1) the region-based STC scheme highly outperforms the STC of \cite{tosample} and 2) the performance of the region-based STC scheme follows closely the ideal performance of ETC, while reducing the computational load in the controller.  
	
	\subsection{Non-Homogeneous System}
	Consider the forced Van der Pol oscillator:
	\begin{equation}
		\dot{\zeta}_1(t) = \zeta_2(t), \quad
		\dot{\zeta}_2(t) = (1-\zeta^2_1(t))\zeta_2(t) -\zeta_1(t) + \upsilon(t),
		\label{vanderpol}
	\end{equation}
	with $\upsilon(t) = -\zeta_2(t)-(1-\zeta^2_1(t))\zeta_2(t)$. Assuming an ETC implementation, and homogenizing the system with an auxilliary variable $w$, according to the procedure presented in \cite{anta2012isochrony} (Lemma IV.4 therein), the extended system \eqref{etc system} becomes:
	\begin{equation}
		\dot{\xi} = 
		\begin{bmatrix}
			\xi_2w^2\\
			(w^2-\xi^2_1)\xi_2 - \xi_1w^2 -\epsilon_2w^2-(w^2-\epsilon_1^2)\epsilon_2\\
			0\\
			-\xi_2w^2\\
			-(w^2-\xi^2_1)\xi_2 + \xi_1w^2 +\epsilon_2w^2+(w^2-\epsilon_1^2)\epsilon_2\\
			0
		\end{bmatrix}
		\label{homogenized system}
	\end{equation}
	where $\xi=[\zeta_1,\zeta_2,w,\varepsilon_1,\varepsilon_2,\varepsilon_w]^{\top}$, $\epsilon_i=\xi_i+\varepsilon_i$, with $\varepsilon$ being the measurement error \eqref{measurement error}. The homogeneity degree of the extended system is $\alpha=2$.
	Observe that the trajectories of the original system \eqref{vanderpol} coincide with the trajectories of \eqref{homogenized system}, if the inital condition for $w$ is $w_0=1$. A triggering function based on the approach of \cite{tabuada2007etc} has been obtained in \cite{postoyan2015framework}:
	\begin{equation*}
		\phi(\zeta(t;x),\varepsilon(t;0))=\phi(\xi(t;x,w_0)) = W(|\varepsilon|) - V(\xi_1,\xi_2),
	\end{equation*}
	where $W(|\varepsilon|)=2.222(\varepsilon_1^2+\varepsilon_2^2)$ and $V(\xi_1,\xi_2) = 0.0058679\xi_1^2 + 0.0040791\xi_1\xi_2 + 0.0063682\xi_2^2$ is a Lyapunov function for the original system. Note, that $\phi(\xi(t;x,w_0))$ is already homogeneous of degree $1$. We fix $\setz = [-0.01,0.01]^3$ and define the following sets: 
	\begin{align*}
		\Phi&=\bigcup\limits_{x_0\in[-0.01,0.01]^2}\{x\in\real^2: W(|x_0-x|)-V(x_1,x_2)\leq0\},\\
		\sete &= \{x_0-x\in\real^2:\text{ }x_0\in[-0.01,0.01]^2, \text{ }x\in\Phi\},\\
		\Xi&=\Phi\times[-0.01,0.01]\times\sete\times\{0\}.
	\end{align*}
	Notice that $\Phi$ is exactly such that for all $x_0\in[-0.01,0.01]^2$: $\phi(\xi(t;x_0,w_0))\leq0 \implies \zeta(t;x_0)\in \Phi$. Then, from the definition of $\sete$ and the observation that $w$ remains constant at all time, it is easily verified that $\setz$ and $\Xi$ are compact, contain the origin and satisfy the requirement of Assumption \ref{assumption 1}.
	
	Let us compare the region-based STC to the ideal performance of ETC. Solving Problem \ref{feasibility problem} for $p=3$, we obtain $\delta_0\approx4.3\cdot 10^{-4}$, $\delta_1=0$ and $\delta_2\approx2.1\cdot 10^{-2}$ and $\delta_3\approx 4\cdot10^{-6}$. To obtain $\mu(x,w,t)$ as in \eqref{mu}, we fix $r=0.009$ and $D=\{x\in\mathbb{R}^3: |x|=r\}$ indeed lies in the interior of $\setz$. The state space is discretized into 267 regions $\reg_i$, with $\tau_{126}=0.1$s and $\tau_i=1.01^{-2}\cdot\tau_{i+1}$. The system is initiated at $x=[-0.3,1.7]^{\top}$, and the simulation duration is 5s. In total, the ETC implementation triggered 114 times, whereas the region-based STC implementation triggered 320 times. This is a much better result than the one presented in the published version \cite{delimpaltadakis_tac}, where region-based STC triggered 1448 times. In \cite{delimpaltadakis_tac}, it was conjectured that the conservatism of region-based STC in this case was due to homogenization and lifting in higher dimensions, but it appears that it was more a matter of finding a better set of parameters $\delta$. 
	
	Figures \ref{stc_times_traj} and \ref{etc_times_traj} demonstrate the evolution of the sampling times of region-based STC and ETC, respectively, along the trajectory. In particular, the curve on the $\zeta_1-\zeta_2$ plane is the trajectory of the system, while the 3D curve above the trajectory is the value of the inter-event time of the corresponding point on the trajectory. The direction of the trajectory is from the blue-colored points to the red-colored points. In Fig. \ref{stc_times_traj} the intervals for which the inter-event time remains constant correspond to segments of the trajectory in which the state vector lies inside one particular region $\reg_i$.
	\begin{figure}[!h]
		\centering
		\includegraphics[width=3in]{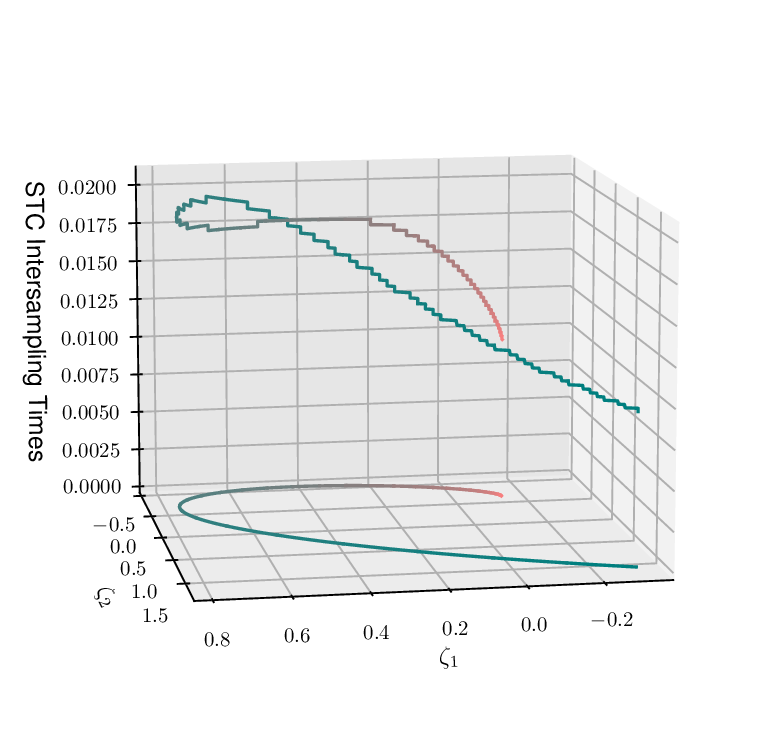}
		\caption{The evolution of region-based STC inter-event times along the trajectory of the forced Van der Pol oscillator.}
		\label{stc_times_traj}
	\end{figure}
	\begin{figure}[!h]
		\centering
		\includegraphics[width=3in]{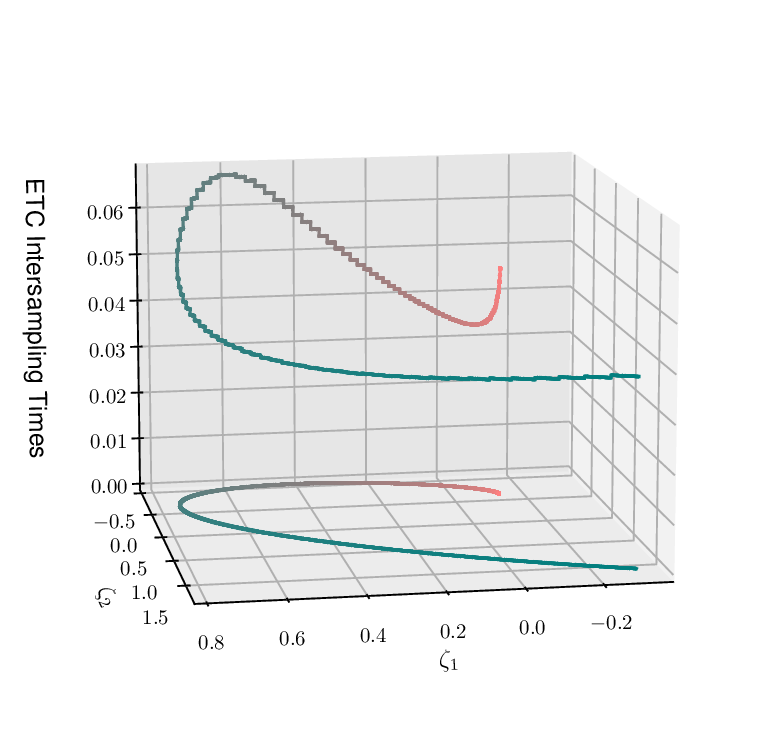}
		\caption{The evolution of ETC inter-event times along the trajectory of the forced Van der Pol oscillator.}
		\label{etc_times_traj}
	\end{figure}
	First, note that in contrast to the previous example, the sampling times do not increase as the system approaches the origin, since the system is not homogeneous and the scaling property \eqref{trig_cond_scaling} does not apply here, i.e. $\phi(\zeta(t;\lambda x))=\phi(\xi(t;\lambda x, 1)) \neq \lambda^{\theta+1}\phi(\xi(\lambda^{\alpha}t;x,1))=\lambda^{\theta+1}\phi(\zeta(\lambda^{\alpha}t;x))$. In fact, as stated in \cite{anta2012isochrony}, the scaling law that applies is : 
	\begin{equation}
		\phi(\xi(t;\lambda x, \lambda w)) = \lambda^{\theta+1}\phi(\xi(\lambda^{\alpha}t;x,w)).
		\label{homogenized scaling}
	\end{equation} 
	However, the similarity of the two figures indicates that the sampling times of the region-based STC approximately follow the trend of the ETC sampling times. This indicates that the approximations of isochronous manifolds determined by $\mu(x,w,t)$ preserve the spatial characteristics of the actual isochronous manifolds of \eqref{vanderpol}. Intuitively, the preservation of the spatial characteristics could be attributed to the fact that $\mu(x,w,t)$ also satisfies \eqref{homogenized scaling}, which determines the scaling of the isochronous manifolds of the homogenized system \eqref{homogenized system} along its homogeneous rays. Besides, note that the isochronous manifolds of the original system \eqref{vanderpol} are the intersections of the isochronous manifolds of \eqref{homogenized system} with the $w=1$-plane.
	
	\begin{remark}
		This simulation demonstrates that, as mentioned in Remark \ref{framework remark}, the results presented in this work are transferable to any smooth, not necessarily homogeneous, system. 
	\end{remark}
	\section{Conclusion and Future Work}
	In this work, a region-based STC policy for nonlinear systems that enables a trade-off between online computations and updates was presented. The scheme employs a state-space partition dictated by inner-approximations of isochronous manifolds of nonlinear ETC systems. To derive such approximations, theoretical issues of \cite{anta2012isochrony} have been addressed and a computational algorithm has been proposed. Finally, simulation results have demonstrated the effectiveness of region-based STC.
	
	\section{Acknowledgements}
	The authors would like to thank Cees F. Verdier for assisting to the implementation of the algorithm described in Section \ref{algorithm section}, Adolfo Anta and Paulo Tabuada for fruitful discussions on this work, and Xiangru Xu for pointing out references \cite{comparison1971} and \cite{lemma_counterexample}.
	
	\section*{Appendix}
	To conduct the proofs of the previously presented lemmas and theorems, we first introduce some preliminary concepts.
	
	\subsection{Higher Order Differential Inequalities}
	\begin{definition}[Type $W^*$ functions \cite{comparison1971}]
		The function $g: \mathbb{R}^n\rightarrow\mathbb{R}$ is said to be of type $W^*$ on a set $S\subseteq\mathbb{R}^n$ if $g(x)\leq g(y)$ for all $x,y\in S$ such that $x_n=y_n$, $x_i\leq y_i$ ($i=1,2,...,n-1$), where $x_i,y_i$ denote the $i$-th component of the $x$ and $y$ vector respectively.
	\end{definition}
	\begin{definition}[Right maximal solution \cite{comparison1971}]
		Consider the $p$-th order differential equation:
		\begin{equation} \label{comparison equation}
			u^{(p)}(t)=g(t,u(t),\dot{u}(t),\dots,u^{(p-1)}(t)),
		\end{equation}
		where $u:\mathbb{R}^+\rightarrow\mathbb{R}$ and $g(\cdot)$ is continuous on $[0,T]\times\mathbb{R}^p$. A solution $u_m(t;t_0,U_m)$, where $t_0$ is the initial time instant and $U_m\in\mathbb{R}^p$ is the vector of initial conditions, is called a right maximal solution of \eqref{comparison equation} on an interval $[t_0,\alpha)\subset[0,T]$ if
		\begin{equation*}
			u^{(i)}(t;t_0,U_0)\leq u^{(i)}_m(t;t_0,U_m), \quad t\in[t_0,\alpha)\cap[t_0,\alpha^*),
		\end{equation*}
		for any other solution $u(t;t_0,U_0)$ with initial condition $U_0\preceq U_m$ defined on $[t_0,\alpha^*)$, for all $i=0,1,2,\dots,m-1$. 
	\end{definition}
	\begin{lemma}[Higher Order Comparison Lemma \cite{comparison1971}]
		\label{higher order comparison lemma}
		Consider a system of first order differential equations:
		\begin{equation} \label{dif eq system}
			\dot{\zeta}(t) = f(t, \zeta(t)).
		\end{equation}
		Let $\upsilon:D_r\rightarrow\mathbb{R}$ and let $\upsilon\in \mathcal{C}^p$, $f\in \mathcal{C}^{p-1}$ on $D_r$, where $D_r=\{(t,x)|0\leq t\leq T <+\infty, |x| < r\}$. Let $g(\cdot)$ of \eqref{comparison equation} be of type $W^*$ on $S \subseteq \mathbb{R}^{p+1}$ for each $t$, where
		\begin{equation*}
			\begin{aligned}
				S=\Big\{\Big(t,\upsilon(t,\zeta(t)),\dot{\upsilon}(t,\zeta(t)),\dots,\upsilon^{(p-1)}(t,\zeta(t))\Big)|\text{ }(t,\zeta(t))\in D_r\Big\}
			\end{aligned}
		\end{equation*} and
		\begin{equation*}
			\upsilon^{(i)}(t,\zeta(t)) = \frac{\partial\upsilon^{(i-1)}(t,\zeta(t))}{\partial t} + \frac{\partial\upsilon^{(i-1)}(t,\zeta(t))}{\partial \zeta(t)}\cdot f(t, \zeta(t)).
		\end{equation*}
		Assume that:
		\begin{equation*}
			\upsilon^{(p)}(t,\zeta(t))\leq g(t,\upsilon(t,\zeta(t)),\dot{\upsilon}(t,\zeta(t)),\dots,\upsilon^{(p-1)}(t,\zeta(t))),
		\end{equation*}
		for all $(t,\zeta(t))\in D_r$. Let $J$ denote the maximal interval of existence of the right maximal solution $u_m(t;0,U_m)$ of \eqref{comparison equation}. If $\upsilon^{(i)}(0,\zeta_0)=u^{(i)}_m(0;0,U_m)$ ($i=0,1,2,\dots,p-1$), where $u^{(i)}_m(0;0,U_m)$ are the components of the initial condition $U_m$ of $u_m(t;0,U_m)$, then:
		\begin{equation*}
			\upsilon^{(i)}(t,\zeta(t;0,\zeta_0)) \leq u_m^{(i)}(t;0,U_m), \quad t\in J \cap [0,T],
		\end{equation*}
		for all $i=0,1,2,\dots,p-1$. 
	\end{lemma}
	
	\subsection{Monotone Systems}
	\begin{definition}[Monotone System\cite{monotone2008}]
		Consider a system:
		\begin{equation} \label{monotone}
			\dot{\zeta}(t) = f(\zeta(t)).
		\end{equation}
		The system \eqref{monotone} is called monotone if:\\ \centerline{$\zeta_0 \preceq \zeta_1\implies \zeta(t;t_0,\zeta_0)\preceq\zeta(t;t_0,\zeta_1)$.}
	\end{definition}
	\begin{proposition}[\hspace{1sp}\cite{monotone2008}] \label{monotone systems proposition}
		Consider the system \eqref{monotone}.
		If the off-diagonal entries of the Jacobian $\frac{\partial f}{\partial \zeta}$ are non-negative, then the system \eqref{monotone} is monotone.
	\end{proposition}
	
	\subsection{Technical Proofs}
	\begin{proof}[\textbf{Proof of Theorem \ref{theorem 1}}]
		Define $\tau^{\downarrow}(x) = \inf\{t>0 : \mu(x,t)=0\}$. \eqref{one-zero-crossing of bound} implies that $\mu(x,\tau^{\downarrow}(x))=0$ is the only zero-crossing of $\mu(x,t)$ w.r.t. $t$ for any given $x$. Hence: 
		\begin{equation*}
			\underline{M}_{\tau_{\star}} =\{x\in\mathbb{R}^n: \mu(x,\tau_{\star})=0\}= \{x\in\mathbb{R}^n: \tau^{\downarrow}(x)=\tau_{\star}\},
		\end{equation*}
		Equations \eqref{scaling of bound} and \eqref{one-zero-crossing of bound} imply that $\underline{M}_{\tau_{\star}}$ satisfies \eqref{star manifolds} and \eqref{encirclement} (see Remark \ref{homogeneity remark}).
		
		It is left to prove that $\underline{M}_{\tau_{\star}}$ is an inner approximation of $M_{\tau_{\star}}$. Notice that $\phi(\xi(\tau(x);x))=0$ together with \eqref{bound time validity} and \eqref{bound init cond}, imply that the first zero-crossing of $\mu(x,t)$ happens before the one of the triggering function: 
		\begin{equation}
			\tau^{\downarrow}(x)\leq\tau(x).
			\label{downarrow leq}
		\end{equation}
		Furthermore, \eqref{scaling of bound} implies that $\tau_{\downarrow}(x)$ also satisfies the scaling law \eqref{scaling} (the proof for this argument is the exact same to the one derived in \cite{tosample} for the scaling laws of inter-event times.) The fact that both $\tau_{\downarrow}(x)$ and $\tau(x)$ satisfy \eqref{scaling}, i.e. they are strictly decreasing functions along homogeneous rays, alongside \eqref{downarrow leq} implies that:
		$\tau(x_1)=\tau^{\downarrow}(x_2)=\tau_{\star} \implies |x_1| \geq |x_2|$, for all $x_1$,$x_2$ on a homogeneous ray. Thus, since $\underline{M}_{\tau_{\star}}$ satisfies \eqref{star manifolds}, we get that for all $x \in \underline{M}_{\tau_{\star}}$:
		\begin{equation*}
			\exists! \kappa_x\geq 1 \text{ s.t. } \kappa_x x \in M_{\tau_i} \text{ } \mathrm{and} \text{ } \not\exists \lambda_x\in(0,1) \text{ s.t. } \lambda_x x \in M_{\tau_i}.
		\end{equation*}
		The proof is now complete.	 
	\end{proof}
	
	\begin{proof}[\textbf{Proof of Lemma \ref{our lemma}}]
		Introduce the following linear system:
		\begin{equation} \label{aux system}
			\dot{\chi}=\begin{bmatrix}
				0 &1 &0 &\dots &0 &0\\
				0 &0 &1 &\dots &0 &0\\
				\vdots &\vdots &\qquad &\ddots &\vdots &\vdots\\
				0 &0 &0 &\dots &1 &0\\
				0 &0 &0 &\dots &0 &1\\
				\delta_0 &\delta_1 &\delta_2 &\dots &\delta_{p-2} &\delta_{p-1}\\
			\end{bmatrix}\chi + 
			\begin{bmatrix}
				0\\
				\vdots\\
				0\\
				\delta_p
			\end{bmatrix}.
		\end{equation}
		Notice that \eqref{aux system} represents the $p$-th order differential equation $\chi^{(p)}=\sum_{i=0}^{p-1}\delta_i\chi^{(i)}+\delta_p$.
		The proof makes use of Lemma \ref{higher order comparison lemma}. Using the notation of Lemma \ref{higher order comparison lemma}, we identify:
		\begin{align*}
			&v(t,\xi(t))\equiv\phi(\xi(t)), \quad \forall\xi(t)\in\Omega_d,\\
			&f(t,\xi(t))\equiv F(\xi(t)), \quad \forall \xi(t) \in \Omega_d,\\
			&g(t,v,v',...,v^{(p-1)})\equiv \sum_{i=0}^{p-1}\delta_iv^{(i)} + \delta_p.
		\end{align*}
		
		For $t>\tau_{\xi_0}$, $\xi(t;\xi_0)$ may not belong to $\Omega_d$. Thus, $\upsilon(\cdot)$ is well-defined only in the interval $[0,\tau_{\xi_0})$. Since $\delta_i\geq0$ for all $i$, $g$ is of type $W^*$ in $\mathbb{R}^+\times\mathbb{R}^p$. Moreover, it is clear that $v\in C^p$ and $f \in C^{p-1}$ on $[0,\tau_{\xi_0})\times\Omega_d$. Inequality \eqref{delta ineq} translates to $v^{(p)}(t,z) \leq g(t,v,v',...,v^{(p-1)})$
		for $(t,z)\in[0,\tau_{\xi_0})\times\Omega_d$. 
		
		Furthermore, according to Proposition \ref{monotone systems proposition}, the linear system \eqref{aux system} is monotone, since all off-diagonal entries of its jacobian are non-negative ($\delta_i\geq0$ for all $i$). This implies that any solution of \eqref{aux system} is a right maximal solution, and its maximal interval of existence is $J=[0,+\infty)$. Consider the solution $\chi(t;X(\xi_0))$, where $X(\xi_0)=\begin{bmatrix}\phi(\xi_0) &\mathcal{L}_F\phi(\xi_0) &\dots &\mathcal{L}_F^{p-1}\phi(\xi_0)\end{bmatrix}^{\top}$. Observe that the components of the initial condition $X(\xi_0)$ and $\mathcal{L}_F^i\phi(z)$ ($i=0,1,2,\dots,p-1$) are equal. All conditions of Lemma \eqref{higher order comparison lemma} are satisfied. Thus, we can conclude that for all $\xi_0\in\Omega_d$:
		\begin{equation*}
			\phi(\xi(t;\xi_0)) \leq \chi_1(t;X(\xi_0)), \quad \forall t \in [0,\tau_{\xi_0}).
		\end{equation*}
		Notice that $\psi_1(y(\xi_0),t)=\chi_1(t;X(\xi_0))$ for all $t$. Hence $\phi(\xi(t;\xi_0)) \leq \psi_1(y(\xi_0),t), \quad \forall t \in [0,\tau_{\xi_0}).$
	\end{proof}
	
	To prove Theorem \ref{main theorem}, we first derive the following results.
	\begin{proposition}
		Consider coefficients $\delta_i$ ($i=0,1,...,p$) solving Problem \ref{feasibility problem}, and define an upper-bound $\psi_1(x,t)$ of the triggering function $\phi(\xi(t;x))$ as dictated in Lemma \ref{our lemma}. Let:
		\begin{equation}
			\eta_1(x,t) := C \boldsymbol{e}^{At}\eta(x,0), 
			\label{eta}
		\end{equation}
		where $A$ is as in \eqref{linear system}, $C = \begin{bmatrix} 1 &0 &\dots &0 \end{bmatrix}$  and:
		\begin{equation}
			\eta(x,0) := 
			\begin{bmatrix}
				\phi\Big((x,0)\Big)\\
				\max\bigg(\mathcal{L}_f\phi\Big((x,0)\Big), 0\bigg)\\
				\vdots\\
				\max\bigg(\mathcal{L}_f^{p-1}\phi\Big((x,0)\Big), 0\bigg)\\
				\delta_p
			\end{bmatrix}.\\
			\label{eta init cond}
		\end{equation}
		The function $\eta_1(x,t)$ satisfies:
		\begin{equation}
			\eta_1(x,t)\geq\phi(\xi(t;x)), \quad \forall t \in [0,\tau(x)] \text{ }\mathrm{and}\text{ } \forall x \in \setz.
			\label{eta bounds phi}
		\end{equation}
		\label{eta bounds phi proposition}
	\end{proposition}
	\begin{proof}
		Notice that  $\eta_1$ is the first component of the solution $\eta(x,t)$ to the same linear dynamical system \eqref{linear system} as $\psi$, with initial condition: $\psi(x,0)\preceq\eta(x,0)$. Since the system \eqref{linear system} is monotone, according to Proposition \ref{monotone systems proposition}, the following holds:
		\begin{equation*}
			\eta_1(x,t) \geq \psi_{1}(x,t) \geq \phi(\xi(t;x)), \quad \forall t \in [0,\tau_{\xi_0}) \text{ }\mathrm{and}\text{ } \forall x \in \setz,
		\end{equation*}
		since $x\in \setz\implies \xi_0=(x,0)\in\Xi\subset\Omega_d$. By the definition of $\Xi$ in Assumption \ref{assumption 1}, $\xi(t;x)\in\Xi$ for all $t\in[0,\tau(x)]$. But $\tau_{\xi_0}$ is defined in \eqref{tau xi0} as the escape time of $\xi(t;x)$ from $\Omega_d$, and $\Xi\subset\Omega_d$; i.e. $\tau(x)<\tau_{\xi_0}$. Thus \eqref{eta bounds phi} is satisfied.
	\end{proof}
	\begin{proposition} \label{eta increasing}
		The function $\eta_1(x,t)$ of \eqref{eta} is strictly increasing w.r.t. $t$ for all $t>0$.
	\end{proposition}
	\begin{proof}
		In the following $\eta_1^{(i)}(x,t)$ denotes the $i$-th derivative of $\eta_1(x,t)$ w.r.t. $t$. At $t=0$, initial condition \eqref{eta init cond} implies that $\eta_1^{(i)}(x,0) \geq 0$ for all $i =1,\dots,p-1$. For $\eta_1^{(p)}(x,0)$:
		\begin{equation*}
			\eta_1^{(p)}(x,0) = \sum_{i=0}^{p-1}\delta_i\eta_{i+1}(x,0)+\delta_p \geq \delta_0\phi\Big((x,0)\Big) + \delta_p > 0,
		\end{equation*}
		since $\eta_{i+1}(x,0)\geq0$ for all $i=0,\dots,p-1$, and \eqref{delta init cond} and \eqref{delta positivity} hold. Differentiating $\eta_1^{(p)}$ w.r.t. $t$, we get:
		\begin{equation*}
			\eta_1^{(p+1)}(x,0)=\sum_{i=0}^{p-1}\delta_i\eta_1^{(i+1)}(x,0)\geq 0.
		\end{equation*} 
		Similarly, $\eta_1^{(i)}(x,0)\geq0$, for all $i$.
		Hence $\eta_1^{(i)}(x,0)\geq0$ for all $i\in\mathbb{N}-\{0\}$, and in particular $\eta_1^{(p)}(x,0) > 0$. This implies that the function $\eta_1(x,t)$ is strictly increasing for all $t>0$.
	\end{proof}
	We are ready to prove Theorem \ref{main theorem}. 
	\begin{proof}[\textbf{Proof of Theorem \ref{main theorem}}]
		First, notice that $\mu(x,t)$ satisfies \eqref{scaling of bound}, by construction.
		Let $D=\{x\in\mathbb{R}^n: |x|=r\}$, with $r>0$ such that $D \subset \setz$. Notice that for $x\in D$: $\mu(x,t)=\eta(x,t)$. Thus, according to Proposition \ref{eta bounds phi proposition} :
		\begin{equation}
			\mu(x,t)=\eta_1(x,t)\geq\phi(\xi(t;x)), \quad \forall t \in [0,\tau(x)] \text{ }\mathrm{and}\text{ } \forall x \in D.
			\label{eta bounds phi D}
		\end{equation}
		Consider now any $x_0\in\mathbb{R}^n-\{0\}$ and a $\lambda>0$ such that $x_D=\lambda x_0\in D$. Employing \eqref{scaling of bound}, \eqref{trig_cond_scaling} and \eqref{eta bounds phi D} we get:
		\begin{equation*}
			\begin{aligned}
				&\mu(x_D,t)\geq\phi(\xi(t;x_D)), \quad \forall t \in [0,\tau(x_D)] \iff\\
				&\mu(x_0,\lambda^{\alpha}t) \geq \phi(\xi(\lambda^{\alpha}t;x_0)), \quad \forall t \in [0,\tau(x_D)] \iff\\
				&\mu(x_0,t) \geq \phi(\xi(t;x_0)), \quad \forall x_0\in\mathbb{R}^n-\{0\} \text { and } t \in [0,\tau(x_0)],
			\end{aligned}
		\end{equation*}
		since $\lambda^{\alpha}\tau(x_D)=\tau(x_0)$. Thus, $\mu(x,t)$ satisfies \eqref{bound time validity}.
		
		It remains to be shown that $\mu(x,t)$ satisfies \eqref{one-zero-crossing of bound}. Notice that $\mu(x,0) = \phi\Big((x,0)\Big) < 0$ for all $x\in\mathbb{R}^n-\{0\}$. Moreover, since \eqref{bound time validity} holds, we get that:
		\begin{equation*}
			\mu(x,\tau(x))\geq\phi(\xi(\tau(x);x))=0.
		\end{equation*}
		From Assumption \ref{assumption 1} we have that such a $\tau(x)$ always exists. Thus, for all $x\in\mathbb{R}^n-\{0\}$ there exists $\tau_{\downarrow}(x)>0$ such that $\mu(x,\tau_{\downarrow}(x))=0$. Moreover, since $\mu(x,t)=\eta(x,t)$ for $x\in D$, then according to Proposition \ref{eta increasing} $\mu(x,t)$ is strictly increasing w.r.t. $t$ for all $t>0$ and for all $x\in D$. Finally, incorporating \eqref{scaling of bound} we get that: $\mu(x,t)$ is strictly increasing w.r.t. $t$ for all $t>0$ and for all $x\in \mathbb{R}^n-\{0\}$; i.e. $\tau_{\downarrow}(x)$ is unique. Thus, $\mu(x,t)$ satisfies \eqref{one-zero-crossing of bound}.
	\end{proof}
	
	\subsection{Non-Homogeneous Systems}
	As stated in Remark \ref{framework remark}, in \cite{anta2012isochrony} a procedure is proposed that renders any smooth nonlinear system homogeneous of degree $\alpha > 0$, by embedding it to higher dimensions and adding an extra variable $w$, with dynamics $\dot{w}=0$. Specifically, a nonlinear system:
	\begin{equation} \label{original system}
		\dot{\zeta}(t)=f(\zeta(t)),
	\end{equation} with $\zeta(t)\in\real^n$ is homogenized as follows:
	\begin{equation}\label{homogenization procedure}
		\begin{bmatrix}
			\dot{\zeta}(t)\\\dot{w}(t)
		\end{bmatrix} = \begin{bmatrix}w^{\alpha+1}f_1(w^{-1}\zeta(t))\\
			w^{\alpha+1}f_2(w^{-1}\zeta(t))\\ \vdots\\ w^{\alpha+1}f_n(w^{-1}\zeta(t))\\0 \end{bmatrix} = \tilde{f}(\zeta(t),w(t)).
	\end{equation}
	Likewise, an ETC system \eqref{etc system} is homogenized by introducing $w$ and the corresponding dummy measurement error $\varepsilon_w$ as:
	\begin{equation} \label{homogenized_etc system}
		\begin{aligned}
			&\dot{\xi}(t)= \begin{bmatrix}
				\dot{\zeta}(t)\\\dot{w}(t)\\\dot{\varepsilon}_\zeta(t)\\\dot{\varepsilon}_w(t)
			\end{bmatrix} = \begin{bmatrix} \tilde{f}(\zeta(t), \varepsilon_z(t),w(t))\\0\\ -\tilde{f}(\zeta(t), \varepsilon_z(t),w(t))\\0 \end{bmatrix}=F(\xi(t)),
		\end{aligned}
	\end{equation}
	where $\tilde{f}(\zeta(t), \varepsilon_z(t),w(t))$ is obtained as in \eqref{homogenization procedure}.
	
	An example of the use of the homogenization procedure is demonstrated in Section VII.B. Similarly, one can homogenize a non-homogeneous triggering function $\phi(\zeta(t;x_0),\varepsilon_\zeta(t;0))$ as: $\tilde{\phi}(\xi(t;x_0,w_0)) = w^{\theta+1}\phi(w^{-1}\zeta(t;x_0),w^{-1}\varepsilon_\zeta(t;0))$. Observe that the trajectories of the original system \eqref{original system} with initial condition $x\in\real^n$ coincide with the ones of \eqref{homogenization procedure} with initial condition $(x,1)\in\real^{n+1}$, i.e. on the hyperplane $w=1$. Hence, the inter-event times of the original system $\tau(x)$ coincide with the inter-event times $\tau\Big((x,1)\Big)$ of \eqref{homogenization procedure}. Consequently, in order to apply the proposed region-based STC scheme to a non-homogeneous nonlinear system, we first homogenize it by embedding it to $\real^{n+1}$, and then derive inner-approximations of isochronous manifolds of the extended system \eqref{homogenization procedure}, by replacing $x$ with $(x,w)$ in \eqref{mu}.
	
	However, a technical detail arises that needs to be emphasized. Most triggering functions that are designed for asymptotic stabilization of the origin (e.g. \cite{tabuada2007etc}) satisfy $\phi\Big((0,0)\Big)=0$. Thus, deriving the function $\mu(x,w,t)$ as in Theorem \ref{main theorem} for the extended system \eqref{homogenization procedure}, results for all points $(0,w)\in\real^{n+1}-\{0\}$ on the $w$-axis in:
	\begin{equation*}
		\mu(0,w,t) = C (\tfrac{|w|}{r})^{\theta+1} \boldsymbol{e}^{A(\frac{|w|}{r})^{\alpha}t} 
		\begin{bmatrix}
			0\\\max\bigg( \mathcal{L}_f\phi\Big((0,0)\Big),0 \bigg)\\
			\vdots\\
			\max\bigg( \mathcal{L}_f^{p-1}\phi\Big((0,0)\Big),0 \bigg)\\\delta_p
		\end{bmatrix}
	\end{equation*}
	This implies that for all these points: $\mu(0,w,t)>0$ for all $t>0$. Hence, the $w$-axis does not belong to any inner-approximation $\underline{M}_{\tau_{\star}}=\{(x,w)\in\real^{n+1}:\mu(x,w,\tau_{\star})=0\}$ of isochronous manifolds. In other words, all inner-approximations $\underline{M}_{\tau_{\star}}$ are punctured by the $w$-axis and obtain a singularity at the origin, as shown in Fig. \ref{torus}.
	\begin{figure}[!h]
		\centering
		\includegraphics[width=2in]{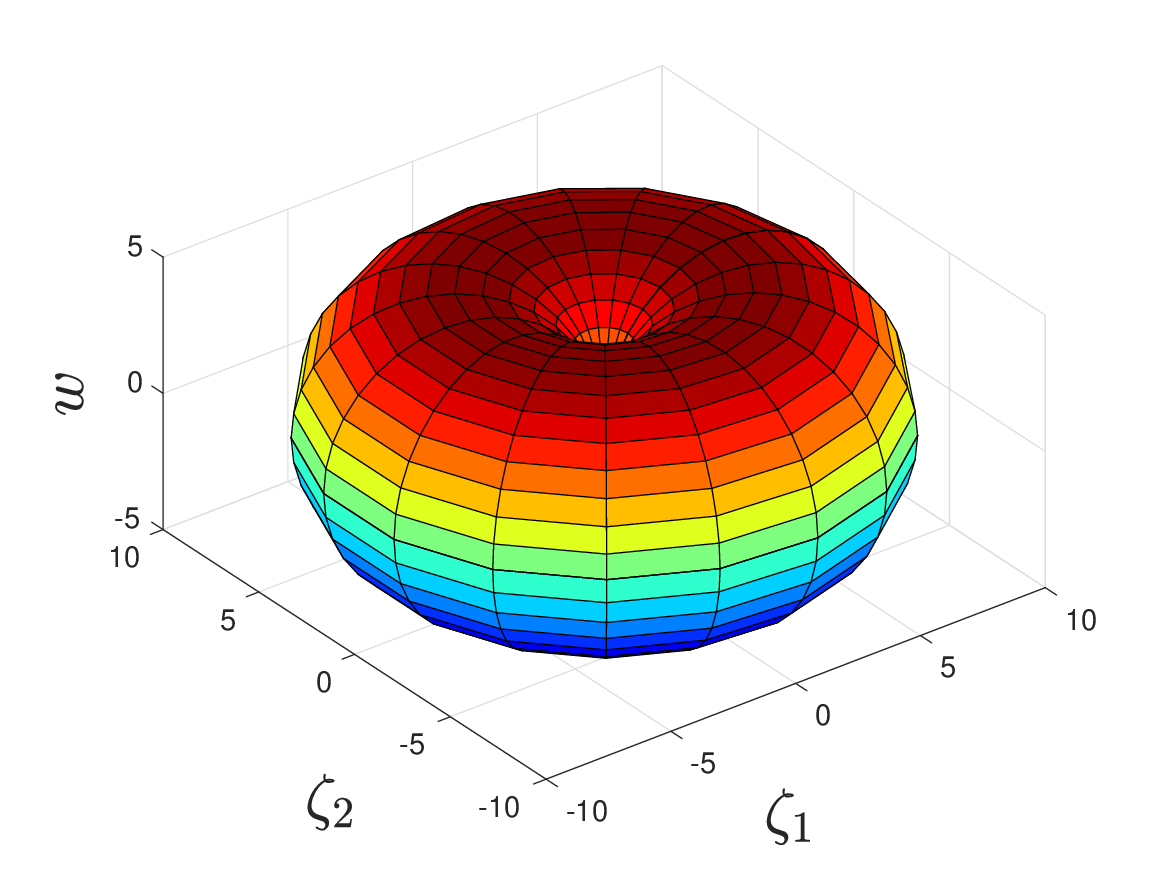}
		\caption{Inner-approximation $\underline{M}_{\tau_{\star}}$ of isochronous manifolds of a homogenized system.}
		\label{torus}
	\end{figure} 
	Consequently, given a finite set of times $\{\tau_1,\dots,\tau_q\}$, discretizing the state-space of the extended system into regions $\reg_i$ delimited by inner-approximations $\underline{M}_{\tau_{i}}$, will always result in a neighbourhood around the $w$-axis not belonging to any region $\reg_i$, as depicted in Fig. \ref{nonhomogeneous_discr}. This implies that a neighbourhood around the origin of the original system \eqref{original system}, which is mapped to a subset of the hyperplane $w=1$ around the $w$-axis in the augmented space $\real^{n+1}$, is not contained to any region $\reg_i$. Thus, no STC inter-event time can be assigned to the points of this neighbourhood.
	\begin{figure}[!h]
		\centering
		\includegraphics[width=2in]{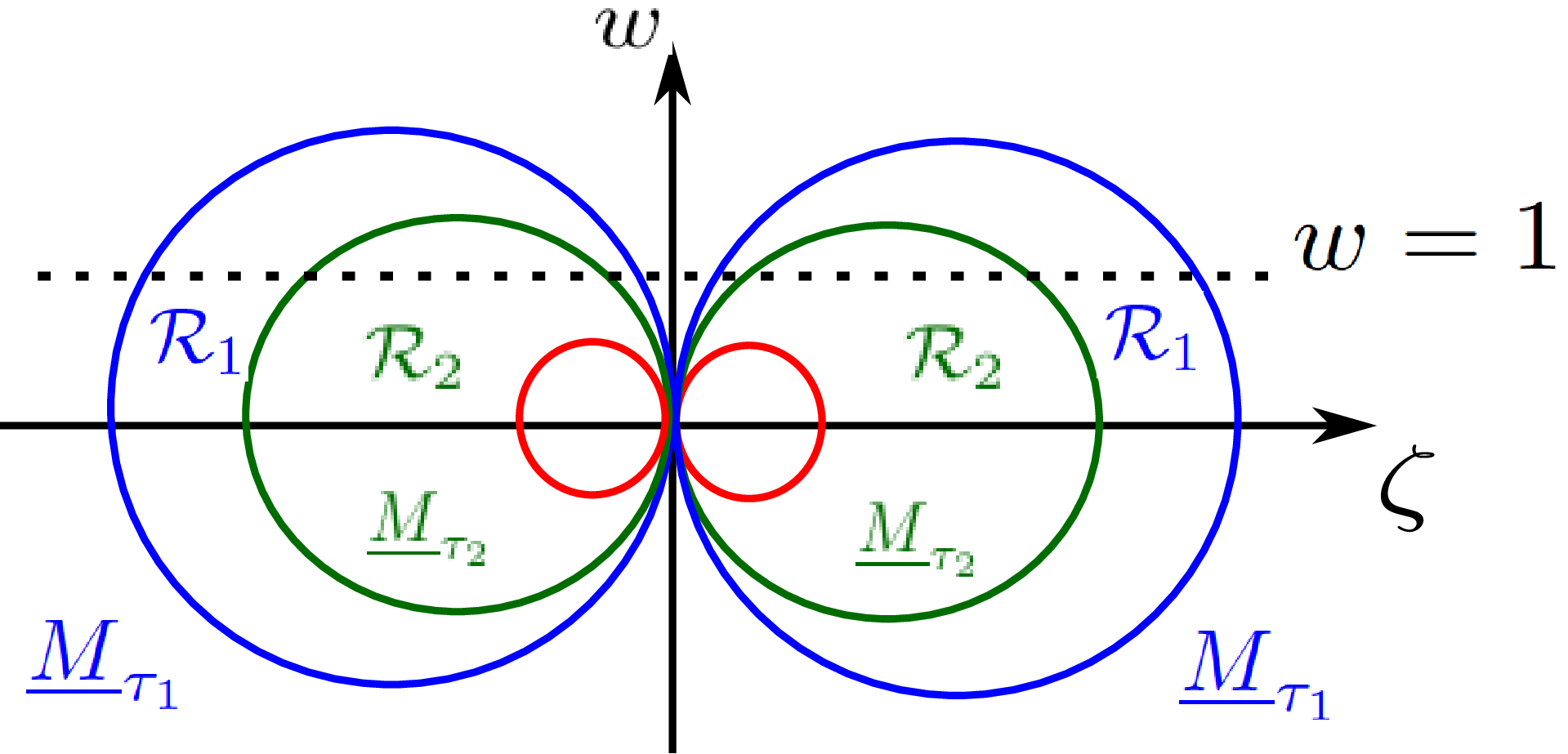}
		\caption{Discretization of the state space of a homogenized system into regions $\reg_i$ delimited by inner-approximations $\underline{M}_{\tau_i}$ (coloured lines) of isochronous manifolds.}
		\label{nonhomogeneous_discr}
	\end{figure} 
	
	However, note that this neighbourhood can be made arbitrarily small, by selecting a sufficiently small time $\tau_1$ for the outermost inner-approximation $\underline{M}_{\tau_1}$. Thus, in order to apply the region-based STC scheme in practice, first we make this neighbourhood arbitrarily small, and then we treat it differently by associating it to a sampling time that can be designed e.g. according to periodic sampling techniques that guarantee stability (e.g. \cite{carnevale2007mati}). In the numerical example of Section VII.B we completely neglect this region, as it was so small that it wasn't even reached during the simulation.
	\begin{remark}\label{last remark}
		Note that, as the $w$-axis acts as a singularity for both the isochronous manifolds $M_{\tau_{\star}}$ (the actual inter-event times there are technically 0, and in practice they could be anything) and their inner-approximations $\underline{M}_{\tau_{\star}}$, the inner-approximations might look very different than the actual manifolds near the $w$-axis.
	\end{remark}
	\begin{remark}\label{other_trig_fun_remark}
		The above technical issue does not arise in cases where $\phi\Big((0,0)\Big)\neq0$. Such an example is the widely used mixed-triggering function $\phi(\xi(t)) = |\varepsilon_{\zeta}(t)|^2-\sigma|\zeta(t)|^2-\epsilon^2$ (e.g. \cite{small_gain_robust_etc}), where $\sigma>0$ is appropriately chosen and $\epsilon>0$.
	\end{remark}	
	\bibliography{bibliography.bib}
	\bibliographystyle{IEEEtran}
\end{document}